\documentclass[11pt,reqno,twoside]{article}

\usepackage{fixltx2e} 

\usepackage{cmap} 

\usepackage[T1]{fontenc}
\usepackage[utf8]{inputenc}
\usepackage{graphicx}
\usepackage{placeins}
\usepackage{enumerate}
\usepackage{amssymb}
\usepackage{algcompatible}
\usepackage{pifont}
\usepackage{subcaption}

\usepackage{verbatim}

\usepackage{setspace}

\let\counterwithin\relax  
\usepackage{lmodern} 
\usepackage[scale=0.88]{tgheros} 

\usepackage{bm} 

\usepackage{amsmath,amsbsy,amsgen,amscd,amsthm,amsfonts,amssymb} 

\usepackage[centering,top=1.5in,bottom=1.2in,left=1in,right=1in]{geometry}

\usepackage{titling}
\usepackage{amssymb}
\graphicspath{ {./images/} }

\usepackage[sf,bf,compact]{titlesec}

\usepackage{booktabs,longtable,tabu} 
\usepackage[font=small,margin=12pt,labelfont={sf,bf},labelsep={space}]{caption}

\usepackage{bbold}
\newcommand{\sumi}{\sum_{i=1}^\infty}

\newcommand{\prodi}{\prod_{i=1}^\infty}
\usepackage[usenames,dvipsnames]{xcolor}

\definecolor{dark-gray}{gray}{0.3}
\definecolor{dkgray}{rgb}{.4,.4,.4}
\definecolor{dkblue}{rgb}{0,0,.5}
\definecolor{medblue}{rgb}{0,0,.75}
\definecolor{rust}{rgb}{0.5,0.1,0.1}

\usepackage{url}
\usepackage[colorlinks=true]{hyperref}
\hypersetup{linkcolor=dkblue}    
\hypersetup{citecolor=rust}      
\hypersetup{urlcolor=rust}     

\usepackage[final]{microtype} 

\newtheoremstyle{myThm} 
    {\topsep}                    
    {\topsep}                    
    {\itshape}                   
    {}                           
    {\sffamily\bfseries}                   
    {.}                          
    {.5em}                       
    {}  

\newtheoremstyle{myRem} 
    {\topsep}                    
    {\topsep}                    
    {}                   
    {}                           
    {\sffamily}                   
    {.}                          
    {.5em}                       
    {}  

\newtheoremstyle{myDef} 
    {\topsep}                    
    {\topsep}                    
    {}                   
    {}                           
    {\sffamily\bfseries}                   
    {.}                          
    {.5em}                       
    {}  

\theoremstyle{myThm}
\newtheorem{theorem}{Theorem}[section]
\newtheorem{lemma}[theorem]{Lemma}
\newtheorem{proposition}[theorem]{Proposition}
\newtheorem{corollary}[theorem]{Corollary}

\theoremstyle{myRem}
\newtheorem{remark}[theorem]{Remark}

\theoremstyle{myDef}

 \newtheorem{example}[theorem]{Example}

\usepackage{fancyhdr}
\usepackage{algorithm}
\usepackage{algpseudocode}

\let\originalleft\left
\let\originalright\right
\renewcommand{\left}{\mathopen{}\mathclose\bgroup\originalleft}
\renewcommand{\right}{\aftergroup\egroup\originalright}

\newcommand{\A}{A}

\newcommand{\muis}{\mu}

\usepackage{mathtools}
\mathtoolsset{centercolon}  

\renewcommand{\phi}{\varphi}

\newcommand{\logeq}{\mathrel{\Leftrightarrow}}

\providecommand{\mathbbm}{\mathbb} 

\newcommand{\R}{\mathbbm{R}}

\newcommand{\N}{\mathbbm{N}}

\newcommand{\bigO}{\mathcal{O}}
\newcommand{\E}{\mathbb{E}}

\definecolor{mygreen}{rgb}{0.1,0.75,0.2}

\newcommand{\post}{\pi}

\newcommand{\Prob}{\operatorname{\mathbbm{P}}}
\newcommand{\Expect}{\operatorname{\mathbb{E}}}
\newcommand{\V}{\operatorname{\mathbb{V}}}

\newcommand{\g}{\mathsf{g}}

\newcommand{\As}{A_{\mbox {\tiny{\rm std}}}}
\newcommand{\Ao}{A_{\mbox {\tiny{\rm opt}}}}

\newcommand{\Gammas}{\Gamma_{\mbox {\tiny{\rm std}}}}
\newcommand{\Gammao}{\Gamma_{\mbox {\tiny{\rm opt}}}}

\newcommand{\pis}{\pi_{\mbox {\tiny{\rm std}}}}
\newcommand{\pio}{\pi_{\mbox {\tiny{\rm opt}}}}

\newcommand{\mus}{\mu_{\mbox {\tiny{\rm std}}}}
\newcommand{\muo}{\mu_{\mbox {\tiny{\rm opt}}}}

\newcommand{\Ks}{K_{\mbox {\tiny{\rm std}}}}
\newcommand{\Ko}{K_{\mbox {\tiny{\rm opt}}}}

\newcommand{\Ss}{S_{\mbox {\tiny{\rm std}}}}
\newcommand{\So}{S_{\mbox {\tiny{\rm opt}}}}

\newcommand{\lambdas}{\lambda_{\mbox {\tiny{\rm std}}}}
\newcommand{\lambdao}{\lambda_{\mbox {\tiny{\rm opt}}}}

\newcommand{\rhos}{\rho_{\mbox {\tiny{\rm std}}}}
\newcommand{\rhoo}{\rho_{\mbox {\tiny{\rm opt}}}}

\newcommand{\Sigmas}{\Sigma_{\mbox {\tiny{\rm std}}}}
\newcommand{\Sigmao}{\Sigma_{\mbox {\tiny{\rm opt}}}}

\newcommand{\dchi}{d_{\mbox {\tiny{$ \chi^2$}}}}

\newcommand{\Nc}{\mathcal{N}}

\usepackage[]{algorithm}

\usepackage{graphicx}

\usepackage{soul}
\usepackage{authblk}
\usepackage[square,numbers]{natbib}
\usepackage{chngcntr}
\usepackage{mathrsfs}
\counterwithin{table}{section}
\counterwithin{algorithm}{section}

\title{Bayesian Update with Importance Sampling: \\ Required Sample Size} 
\author{Daniel Sanz-Alonso and Zijian Wang}

\vspace{.25in} 

\date{University of Chicago}

\makeatletter\@addtoreset{section}{part}\makeatother%
\numberwithin{equation}{section}

\newcommand{\upperRomannumeral}[1]{\uppercase\expandafter{\romannumeral#1}}

\begin{document}
\maketitle 

\begin{abstract}

Importance sampling is used to approximate Bayes' rule in many computational approaches to Bayesian inverse problems, data assimilation and machine learning. This paper reviews and further investigates the required sample size for importance sampling in terms of the $\chi^2$-divergence between target and proposal. We develop general abstract theory and illustrate through numerous examples the roles that dimension, noise-level and other model parameters play in approximating the Bayesian update with importance sampling. Our examples also facilitate a new direct comparison of standard and optimal proposals for particle filtering. 
\end{abstract}

\section{Introduction}
Importance sampling is a mechanism to approximate expectations with respect to a \emph{target} distribution using independent weighted samples from a \emph{proposal} distribution. The variance of the weights  ---quantified by the $\chi^2$-divergence between target and proposal--- gives both necessary and sufficient conditions on the sample size to achieve a desired worst-case error over large classes of test functions. This paper contributes to the understanding of importance sampling to approximate the Bayesian update, where the target is a posterior distribution obtained by conditioning the proposal to observed data. We consider illustrative examples where the $\chi^2$-divergence between target and proposal admits a closed formula and it is hence possible to characterize explicitly the required sample size. These examples showcase the fundamental challenges that importance sampling encounters in high dimension and small noise regimes where target and proposal are far apart. They also facilitate a direct comparison of standard and optimal proposals for particle filtering. 

We denote the target distribution by $\mu$ and the proposal by $\pi$ and assume that both are probability distributions in Euclidean space $\R^d$. We further suppose that the target is absolutely continuous with respect to the proposal, and denote by $g$ the \emph{unnormalized} density between target and proposal so that, for any suitable test function $\varphi,$ 
\begin{align}\label{eq:abscont}
\int_{\R^d} \varphi(u) \mu(du) = \frac{\int_{\R^d}\varphi(u) g(u)  \pi(du)}{\int_{\R^d} g(u) \pi(du)}.
\end{align}
We write this succinctly as $\mu( \varphi) = \pi( \varphi g)/\pi(g).$ Importance sampling approximates $\mu(\varphi)$ using independent samples $\{u^{(n)} \}_{n=1}^N $   from the proposal $\pi, $ computing the numerator and denominator in \eqref{eq:abscont} by Monte Carlo integration,
\begin{align}\label{eq:IS}
\begin{split}
\mu(\phi) &\approx  \frac{\frac1N \sum_{n=1}^N  \varphi(u^{(n)}) g(u^{(n)})}{ \frac{1}{N} \sum_{n=1}^N g(u^{(n)})  } \\
&=\sum_{n =1}^N w^{(n)} \varphi(u^{(n)}), \quad \quad w^{(n)} := \frac{ g(u^{(n)})}  { \sum_{\ell=1}^N g(u^{(\ell)})}.
\end{split}
\end{align}
The weights $w^{(n)}$  ---called \emph{autonormalized} or self-normalized since they add up to one---  can be computed as long as the unnormalized density $g$ can be evaluated point-wise; knowledge of the normalizing constant $\pi(g)$ is not needed. We write \eqref{eq:IS} briefly as $\mu(\phi) \approx \mu^N(\phi)$, where $\mu^N$ is the \emph{random} autonormalized particle approximation measure
 \begin{equation}\label{eq:particleapproxmeasure}
\muis^N := \sum_{n=1}^N w^{(n)} \delta_{u^{(n)}}, \quad \quad u^{(n)} \stackrel{\text{i.i.d.}}{\sim} \pi.
\end{equation}

This paper is concerned with the study of importance sampling in Bayesian formulations to inverse problems, data assimilation and machine learning tasks \cite{agapiou2017importance,sanzstuarttaeb,barber2012bayesian,trillos2020consistency,trillos2018bayesian}, where the relationship $\mu(du) \propto g(u) \pi(du)$ arises from application of Bayes' rule $\Prob(u|y) \propto \Prob(y|u) \Prob(u)$; 
we  interpret $u \in \R^d$ as a parameter of interest,  $\pi \equiv \Prob(u)$ as a prior distribution on $u$, $g(u) \equiv g(u;y) \equiv \Prob(y | u)$ as a likelihood function which tacitly depends on observed data $y \in \R^k$, and $\mu \equiv \Prob(u|y)$ as the posterior distribution of $u$ given $y.$  With this interpretation and terminology, the goal of importance sampling is to approximate posterior expectations using prior samples. Since the prior has fatter tails than the posterior, the Bayesian setting poses further structure into the analysis of importance sampling. In addition, there are several specific features of the application of importance sampling in Bayesian inverse problems, data assimilation and machine learning that shape our presentation and results. 

First, Bayesian formulations have the potential to provide uncertainty quantification by computing \emph{several} posterior quantiles. This motivates considering a worst-case error analysis \cite{dick2013high} of importance sampling over large classes of test functions $\varphi$ or, equivalently, bounding a certain distance between the random particle approximation measure $\mu^N$ and the target $\mu,$ see \cite{agapiou2017importance}.
As we will review in Section \ref{sec:2}, a key quantity in controlling the error of importance sampling with bounded test functions is the $\chi^2$-divergence between target and proposal, given by
$$\dchi(\mu \| \pi) = \frac{\pi(g^2)}{ \pi(g)^2} - 1.$$

Second, importance sampling in inverse problems, data assimilation and machine learning applications is often used as a building block of more sophisticated computational methods, and in such a case there may be little or no freedom in the choice of proposal. For this reason, throughout this paper we view both target and proposal as given and we focus on investigating the required sample size for accurate importance sampling with bounded test functions, following a similar perspective as \cite{CP15,agapiou2017importance,sanz2018importance}.
The complementary question of how to choose the proposal to achieve a small variance for a given test function is not considered here. This latter question is of central interest in the simulation of rare events \cite{rubino2009rare} and has been widely studied since the introduction of importance sampling in \cite{kahn1953methods,kahn1955use},  leading to a plethora of adaptive importance sampling schemes \cite{bugallo2017adaptive}. 

Third, high dimensional and small noise settings are standard in inverse problems, data assimilation and machine learning, and it is essential to understand the scalability of sampling algorithms in these challenging regimes. The curse of dimension of importance sampling has been extensively investigated \cite{BBL08,BLB08,snyder2008obstacles,RH13,chorin2013conditions,agapiou2017importance}. The early works \cite{BBL08,BLB08} demonstrated a \emph{weight collapse} phenomenon, by which unless the number of samples is scaled exponentially with the dimension of the parameter, the maximum weight converges to one. 
The paper \cite{agapiou2017importance} also considered small noise limits and further emphasized the need to define precisely the dimension of learning problems. Indeed, while many inverse problems, data assimilation models and machine learning tasks are defined in terms of millions of parameters, their intrinsic dimension is often substantially lower since $(i)$ all parameters are typically not equally important; $(ii)$ substantial a priori information about some parameters may be available; and  $(iii)$ the data may be lower dimensional than the parameter space. Here we will provide a unified and accessible understanding of the roles that dimension, noise-level and other model parameters play in approximating the Bayesian update. We will do so through examples where it is possible to compute explicitly the $\chi^2$-divergence between target and proposal, and hence the required sample size. 

Finally, in the Bayesian context the normalizing constant $\pi(g)$ represents the marginal likelihood and is often computationally intractable. This motivates our focus on the  \emph{auto-normalized} importance sampling estimator in \eqref{eq:IS}, which estimates \emph{both} $\pi(g\varphi)$ and $\pi(g)$ using Monte Carlo integration, as opposed to unnormalized variants of importance sampling 
\cite{sanz2018importance}.

\subsection*{Main Goals, Specific Contributions and Outline}
The main goal of this paper is to provide a rich and unified understanding of the use of importance sampling to approximate the Bayesian update, while keeping the presentation accessible to a large audience. In Section \ref{sec:2} we investigate the required sample size for importance sampling in terms of the $\chi^2$-divergence between target and proposal. Section \ref{sec:3}  builds on the results in Section \ref{sec:2} to illustrate through numerous examples the fundamental challenges that importance sampling encounters when approximating the Bayesian update in small noise and high dimensional settings. In Section \ref{sec:4} we show how our concrete examples facilitate a new direct comparison of standard and optimal proposals for particle filtering. These examples also allow us to identify model problems where the advantage of the optimal proposal over the standard one can be dramatic. 

Next, we provide further details on the specific contributions of each section and link them to the literature. We refer to \cite{agapiou2017importance} for a more exhaustive literature review. 
\begin{itemize}
\item Section \ref{sec:2} provides a unified perspective on the sufficiency and necessity of having a sample size of the order of the $\chi^2$-divergence between target and proposal to guarantee accurate importance sampling with bounded test functions. 
Our analysis and presentation are informed by the specific features that shape the use of importance sampling to approximate Bayes' rule.
The key role of the second moment of the $\chi^2$-divergence has long been acknowledged \cite{liu1996metropolized,pitt1999filtering}, and it is intimately related to an effective sample size used by practitioners to monitor the performance of importance sampling \cite{kong1992note,kong1994sequential}. A topic of recent interest is the development of adaptive  importance sampling schemes where the proposal is chosen by minimizing  ---over some admissible family of distributions--- the $\chi^2$-divergence with respect to the target \cite{ryu2014adaptive,akyildiz2019convergence}.  The main original contributions of Section \ref{sec:2}  are Proposition \ref{prop:necessary sample size} and Theorem \ref{thm:nece suff summary}, which demonstrate the \emph{necessity} of suitably increasing the sample size with the $\chi^2$-divergence along singular limit regimes. The idea of Proposition \ref{prop:necessary sample size} is inspired by \cite{CP15}, but adapted here from relative entropy to $\chi^2$-divergence. Our results  complement sufficient conditions on the sample size derived in \cite{agapiou2017importance} and necessary conditions for \emph{unnormalized} (as opposed to autonormalized) importance sampling in
 \cite{sanz2018importance}. 
\item In Section \ref{sec:3}, Proposition \ref{prop:inverse_problem_chi2} gives a closed formula for the $\chi^2$-divergence between posterior and prior in a linear-Gaussian Bayesian inverse problem setting. This formula allows us to investigate the scaling of the $\chi^2$-divergence (and thereby the rate at which the sample size needs to grow) in several singular limit regimes, including small observation noise, large prior covariance and large dimension. Numerical examples motivate and complement the theoretical results. In an infinite dimensional setting, Corollary \ref{cor: dimensionality 3 equivalence} 
establishes an equivalence between absolute continuity, finite $\chi^2$-divergence and finite intrinsic dimension. A similar result was proved in more generality in \cite{agapiou2017importance} using the advanced theory of Gaussian measures in Hilbert space \cite{bogachev1998gaussian}; our presentation and proof here are elementary, while still giving the same degree of understanding. 
\item In Section \ref{sec:4} we follow  \cite{BBL08,BLB08,snyder2008obstacles,snyder2015performance,agapiou2017importance} and investigate the use of importance sampling to approximate Bayes' rule within one filtering step in a linear-Gaussian setting. We build on the examples and results in Section \ref{sec:3} to identify model regimes where the performance of standard and optimal proposals can be dramatically different. We refer to  \cite{doucet2001introduction,sanzstuarttaeb} for an introduction to standard and optimal proposals for particle filtering, and to \cite{del2004feynman} for a more advanced presentation. The main original contribution of this section is Theorem \ref{prop: part filt chi2}, which gives a direct comparison of the $\chi^2$-divergence between target and standard/optimal proposals. This result improves on \cite{agapiou2017importance}, where only a comparison between the intrinsic dimension was established. 
\end{itemize}


\section{Importance Sampling and \texorpdfstring{$\chi^2$}{chi2}-divergence}\label{sec:2}
The aim of this section is to demonstrate the central role of the $\chi^2$-divergence between target and proposal in determining the accuracy of importance sampling. In Subsection \ref{ssec:sufficientnecessary} we show how the $\chi^2$-divergence arises in both sufficient and necessary conditions on the sample size for accurate importance sampling with bounded test functions.  Subsection \ref{ssec:ESS} describes a well-known connection between the effective sample size and the $\chi^2$-divergence.  Our investigation of importance sampling to approximate the Bayesian update ---developed in Sections \ref{sec:3} and \ref{sec:4}--- will make use of a closed formula for the $\chi^2$-divergence between Gaussians, which we include in Subsection \ref{ssec:chiGaussians} for later reference.

\subsection{Sufficient and Necessary Sample Size}\label{ssec:sufficientnecessary}
Here we provide general sufficient and necessary conditions on the sample size in terms of 
$$\rho := \dchi(\mu \|\pi)+1.$$
We first review upper-bounds on the worst-case bias and mean-squared error of importance sampling with bounded test functions, which imply that accurate importance sampling is guaranteed if $N \gg \rho$. The proofs can be found in \cite{agapiou2017importance,sanzstuarttaeb} and are therefore omitted. 
	\begin{proposition}[Sufficient Sample Size]\label{prop:sufficient sample size}
		It holds that
				\begin{align*}
				\sup_{|\phi|_\infty\leq 1} \Bigl|\Expect\left[\muis^N(\phi)-\mu(\phi)\right] \Bigr| &\leq\frac{4}{N} \rho, \\  
				\sup_{|\phi|_\infty\leq 1}\Expect \left[ \bigl(\muis^N(\phi)-\mu(\phi)\bigr)^2 \right] &\leq\frac{4}{N} \rho.
				\end{align*}
	\end{proposition}
	
The next result shows the existence of bounded test functions for which the error may be large with a high probability if $N \ll \rho.$ The idea is taken from \cite{CP15}, but we adapt it here to obtain a result in terms of the $\chi^2$-divergence rather than relative entropy. We denote by $\g :=g/\pi(g)$ the \emph{normalized} density between $\mu$ and $\pi,$ and note that $\rho=\pi(\g^2)=\mu(\g).$
	\begin{proposition}[Necessary Sample Size]\label{prop:necessary sample size}
	Let $U \sim \mu.$ For any $N\ge 1$ and $\alpha \in (0,1) $ there exists a test function $\phi$ with $|\phi|_\infty \le 1$ such that 
	\begin{equation}
	\Prob\Bigl(   | \muis^N(\phi) - \mu(\phi) | = \Prob( \g(U) > \alpha \rho \bigr) \Bigr)  \ge 1 - \frac{N} {\alpha\rho}.
	\end{equation}
    \end{proposition}

    \begin{proof}
    Observe that for the test function $\phi(u) :=\mathbbm{1}\{\g(u)\leq\alpha\rho\}$, we have $\mu(\phi) = \Prob\bigl( \g(U) \le \alpha \rho\bigr).$ On the other hand, $\mu^N(\phi)=1$ if and only if $\g(u^{(n)})\leq\alpha\rho$ for all $1 \le n \le N$. This implies that
        \begin{align}\label{eq:boundnec}
        \Prob\left(|\mu^N(\phi)-\mu(\phi)|
        =\Prob(\g(U)>\alpha\rho)\right)
        \geq 1-N\Prob( \g(u^{(1)})>\alpha\rho)
        \geq 1-\frac{N}{\alpha\rho}.
        \end{align}
    \end{proof}
    
The power of Proposition \ref{prop:necessary sample size} is due to the fact that in some singular limit regimes the distribution of $\g(U)$ concentrates around its expected value $\rho.$ In such a case, for any fixed $\alpha\in(0,1)$ the probability of the event $\g(U) > \alpha \rho$ will not vanish as the singular limit is approached.  This idea will become clear in the proof of Theorem \ref{thm:nece suff summary} below. 
    
    In Sections \ref{sec:3} and \ref{sec:4} we will investigate the required sample size for importance sampling approximation of the Bayesian update in various singular limits, where target and proposal become further apart as a result of reducing the observation noise, increasing the prior uncertainty, or increasing the dimension of the problem. To formalize the discussion in a general abstract setting, let $\{ (\mu_\theta, \pi_\theta)\}_{\theta>0}$ be a family of targets and proposals such that $\rho_\theta:= \dchi(\mu_\theta \| \pi_\theta) \to \infty$ as $\theta\to \infty.$ The parameter $\theta$ may represent for instance the size of the precision of the observation noise, the size of the prior covariance, or a suitable notion of dimension. Our next result shows a clear dichotomy in the performance of importance sampling along the singular limit depending on whether  the sample size grows sublinearly or superlinearly with $\rho_\theta.$

	\begin{theorem}\label{thm:nece suff summary}
	Suppose that $\rho_\theta \to \infty$ and that $\mathcal{V}:=\sup_{\theta}\frac{\V [\g_\theta(U_\theta)]}{\rho_\theta^2}<1.$
	Let $\delta >0.$ 
	\begin{enumerate}[(i)]
	\item If $N_\theta = \rho_\theta^{1 + \delta},$ then 
    	\begin{equation}\label{eq:lim var}
    	\lim_{\theta \to \infty} \sup_{|\phi|_\infty\leq 1}\Expect \bigl[ \bigl(\mu_\theta^{N_\theta}(\phi)-\mu_\theta(\phi)\bigr)^2 \bigr]  = 0.
    	\end{equation}
    \item  If $N_\theta = \rho_\theta^{1 - \delta}$, then there exists a fixed $c\in(0,1)$ such that
    	\begin{equation}\label{eq:boundlargetheta}
       \lim_{\theta \to \infty} \sup_{|\phi|_\infty \le 1} \Prob\Bigl( | \mu_\theta^{N_\theta}(\phi) - \mu_\theta(\phi) | > c    \Bigr)  = 1.
    	\end{equation}
	\end{enumerate}
	\end{theorem}
    \begin{proof}
   The proof of $(i)$ follows directly from Proposition \ref{prop:sufficient sample size}. For $(ii)$ we  fix $\alpha\in(0,1- \mathcal{V})$ and $c\in \Bigl(0,1-\frac{\mathcal{V}}{(1-\alpha)^2} \Bigr)$. Let $\phi_\theta( u):=\mathbbm{1}(\g_\theta( u)\leq\alpha\rho_\theta)$ as in the proof of Proposition \ref{prop:necessary sample size}. Then,
        \begin{align*}
        \Prob \bigl(\g_\theta(U_\theta)>\alpha\rho_\theta \bigr)
        \ge  1-\Prob \bigl(|\rho_\theta-\g_\theta(U_\theta)|
            \geq (1-\alpha)\rho_\theta \bigr)
        \geq 1-\frac{\V[\g_\theta(U_\theta) ]}{(1-\alpha)^2\rho_\theta^2}
        \geq 1-\frac{\mathcal{V}}{(1-\alpha)^2} > c.
        \end{align*}
    The bound in \eqref{eq:boundnec} implies that
    \[\Prob\Bigl( | \mu_\theta^{N_\theta}(\phi_\theta) - \mu_\theta(\phi_\theta) | > c \Bigr)
    \geq \Prob\Bigl(   | \muis^N_\theta(\phi_\theta) - \mu_\theta(\phi_\theta) | = \Prob( \g_\theta(U_\theta) > \alpha \rho_\theta \bigr) \Bigr)  
    \ge 1 - \frac{N_\theta} {\alpha\rho_\theta}.\]
    This completes the proof, since  if $N_\theta = \rho_\theta^{1-\delta}$ the right-hand side goes to $1$ as $\theta\to\infty$.  
    \end{proof} 

The assumption that $\mathcal{V}<1$ can be verified for some singular limits of interest, in particular for small noise and large prior covariance limits studied in Sections \ref{sec:3} and \ref{sec:4}; details will be given in Example \ref{example: 1-d gaussian, verify assumption of sup var/rho2}. While the assumption $\mathcal{V} <1$ may fail to hold in high dimensional singular limit regimes, the works \cite{BBL08,BLB08} and our numerical example in Subsection \ref{ssec:dalargedimension} provide compelling evidence of the need to suitably scale $N$ with $\rho$ along those singular limits in order to avoid a weigh-collapse phenomenon. Further theoretical evidence was given for unnormalized importance sampling in \cite{sanz2018importance}.

\subsection{ \texorpdfstring{$\chi^2$}{chi2}-divergence and Effective Sample Size}\label{ssec:ESS}
The previous subsection provides theoretical non-asymptotic and asymptotic evidence that a sample size larger than $\rho$ is necessary and sufficient for accurate importance sampling. Here we recall a well known connection between the $\chi^2$-divergence and the effective sample size
\begin{equation}
\text{ESS} := \frac{1}{\sum_{n=1}^N (w^{(n)} )^2},
\end{equation}
widely used by practitioners to monitor the performance of importance sampling. Note that always  $1 \le \text{ESS} \le N$; it is intuitive that $\text{ESS} = 1$ if the maximum weight is one and $\text{ESS} = N$ if the maximum weight is $1/N.$ To see the connection between $\text{ESS}$ and $\rho$, note that 
\begin{align*}
\frac{\text{ESS}}{N} 
&= \frac{1}{N \sum_{n=1}^N (w^{(n)} )^2}  
= \frac{ \Bigl( \sum_{n=1}^N g(u^{(n)})  \Bigr)^2}{N \sum_{n=1}^N g(u^{(n)})^2} 
= \frac{ \biggl(  \frac{1}{N} \sum_{n=1}^N g(u^{(n)}) \biggr)^2}{\frac{1}{N} \sum_{n=1}^N g(u^{(n)})^2}
 \approx \frac{\pi(g)^2}{\pi(g^2)}.
\end{align*}
Therefore, $\text{ESS} \approx N / \rho:$  if the sample-based estimate of $\rho$ is significantly larger than $N$, $\text{ESS}$ will be small which gives a warning sign that a larger sample size $N$ may be needed.

\subsection{\texorpdfstring{$\chi^2$}{chi2}-divergence Between Gaussians}\label{ssec:chiGaussians}
We conclude this section by recalling an analytical expression for the 
$\chi^2$-divergence between Gaussians. In order to make our presentation self-contained, we include a proof in Appendix \ref{appendixA}.
	\begin{proposition}\label{prop:gaussian_chi2}
		Let $\mu=\mathcal{N}(m,C)$ and $\pi=\mathcal{N}(0,\Sigma)$. If $2\Sigma\succ C$, then
		\[ \rho =\frac{|\Sigma|}{\sqrt{|2\Sigma-C||C|}} \exp \Bigl(m'(2\Sigma-C)^{-1}m \Bigr). \]
        Otherwise, $\rho= \infty$.
	\end{proposition}	
It is important to note that non-degenerate Gaussians $\mu=\mathcal{N}(m,C)$ and $\pi=\mathcal{N}(0,\Sigma)$ in $\R^d$ are always equivalent. However,  $\rho =\infty $ unless $2\Sigma \succ C.$ In Sections \ref{sec:3} and \ref{sec:4} we will interpret $\mu$ as a posterior and $\pi$ as a prior, in which case automatically $\Sigma \succ C$ and $\rho <\infty.$

\section{Importance Sampling for Inverse Problems}\label{sec:3}
In this section we study the use of importance sampling in a linear Bayesian inverse problem setting where the target and the proposal represent, respectively, the posterior and the prior distribution.  In Subsection \ref{ssec:ipsetting} we describe our setting and we also derive an explicit formula for the $\chi^2$-divergence between the posterior and the prior. This explicit formula allows us to determine the scaling of the $\chi^2$-divergence in small noise regimes (Subsection \ref{ssec:ipnoise}), in the limit of large prior covariance (Subsection \ref{ssec:ippriorscaling}), and in a high dimensional limit  (Subsection \ref{ssec:ipdimension}). Our overarching goal is to show how the sample size for importance sampling needs to grow along these limiting regimes in order to maintain the same level of accuracy.  

\subsection{Inverse Problem Setting and \texorpdfstring{$\chi^2$}{chi2}-divergence Between Posterior and Prior}\label{ssec:ipsetting}
Let $A\in \R^{k \times d}$ be a given \emph{design} matrix and consider the linear inverse problem of recovering $u \in \R^d$ from data $y \in \R^k$ related by 
\begin{equation}\label{eq:inverseproblem}
y=\A u+\eta, \quad \quad \eta \sim \Nc(0,\Gamma),
\end{equation}
where $\eta$ represents measurement noise. We assume henceforth that we are in the underdetermined case $k \le d$, and that $A$ is full rank. 
We follow a Bayesian perspective and set a Gaussian prior on $u$, $u \sim \pi = \mathcal{N}(0,\Sigma).$ We assume throughout that $\Sigma$ and $\Gamma$ are given symmetric positive definite matrices. The solution to the Bayesian formulation of the inverse problem is the posterior distribution $\mu$ of $u$ given $y.$ 
We are interested in studying the performance of importance sampling with proposal $\pi$ (the prior) and target $\mu$ (the posterior). We recall  that under this linear-Gaussian model the posterior distribution is Gaussian \cite{sanzstuarttaeb}, and we denote it by $\mu = \Nc(m,C)$. In order to characterize the posterior mean $m$ and covariance $C$, we introduce standard data assimilation notation	
	\begin{align*}
S & := A \Sigma A' + \Gamma, \\
K & := \Sigma A' S^{-1},
\end{align*}
where $K$ is the Kalman gain. Then we have
\begin{align}\label{eq:mandC}
\begin{split}
m &= K y, \\
C &= (I-KA) \Sigma.
\end{split}
\end{align}
Proposition \ref{prop:gaussian_chi2} allows us to obtain a closed formula for the quantity $\rho = \dchi(\mu\| \pi) + 1$, noting that \eqref{eq:mandC} implies that
	\begin{align*}
	2 \Sigma - C &= (I + KA) \Sigma \\
	  & = \Sigma + \Sigma A'S^{-1} A \Sigma \succ 0.
	\end{align*}
The proof of the following result is then immediate and therefore omitted. 
		\begin{proposition}\label{prop:inverse_problem_chi2}
		Consider the inverse problem \eqref{eq:inverseproblem} with prior $u\sim \pi = \mathcal{N}(0,\Sigma)$ and posterior $\mu = \Nc(m,C)$ with $m$ and $C$ defined in \eqref{eq:mandC}. Then  $\rho = \dchi( \mu \| \pi) + 1$ admits the explicit characterization
		\[\rho = (|I+K A||I-K A|)^{-\frac{1}{2}} \exp \Bigl(y' K' [ (I + KA) \Sigma]^{-1} Ky \Bigr).\] 
	\end{proposition}
	In the following two subsections we employ this result to derive by direct calculation the rate at which the posterior and prior become further apart ---in $\chi^2$-divergence--- in small noise and large prior regimes. To carry out the analysis we use parameters $\gamma^2, \sigma^2>0$ to scale the noise covariance, $\gamma^2 \Gamma,$ and the prior covariance, $\sigma^2 \Sigma.$

%

\subsection{Importance Sampling in Small Noise Regime}\label{ssec:ipnoise}
To illustrate the behavior of importance sampling in small noise regimes, we first introduce a motivating numerical study.  A similar numerical setup was used in \cite{BBL08} to demonstrate the curse of dimension of importance sampling. We consider the inverse problem setting in Equation \eqref{eq:inverseproblem} with $d = k = 5$ and noise covariance $\gamma^2 \Gamma.$ We conduct $18$ numerical experiments with a fixed data $y$. For each experiment, we perform importance sampling $400$ times, and report in Figure \ref{fig:Noise Scaling d=5} a histogram with the largest autonormalized weight in each of the $400$ realizations. The $18$ experiments differ in the sample size $N$ and the size of the observation noise $\gamma^2.$ In both Figures \ref{fig:Noise Scaling d=5}.a and \ref{fig:Noise Scaling d=5}.b we consider three choices of $N$ (rows) and three choices of $\gamma^2$ (columns). These choices are made so that in Figure \ref{fig:Noise Scaling d=5}.a it holds that $N = \gamma^{-4}$ along the bottom-left to top-right diagonal, while in Figure \ref{fig:Noise Scaling d=5}.b $N = \gamma^{-6}$ along the same diagonal.

\begin{figure}[!htb] 
    \begin{minipage}{.5\textwidth}
       
        \includegraphics[width=1\linewidth, height=0.265\textheight]{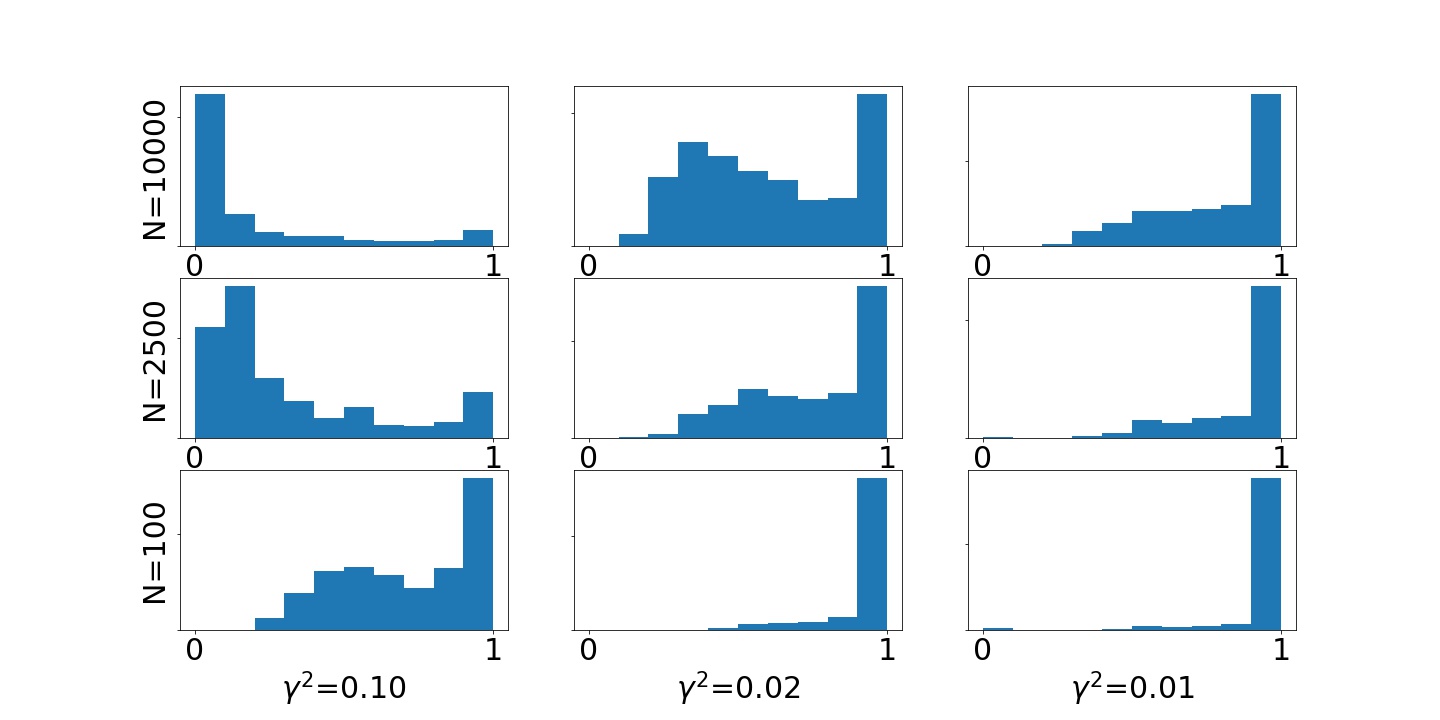}
        \subcaption{$N=\gamma^{-4}$.}
        \label{fig:prob1_6_2}
    \end{minipage}%
    \begin{minipage}{0.5\textwidth}
       
        \includegraphics[width=1\linewidth, height=0.265\textheight]{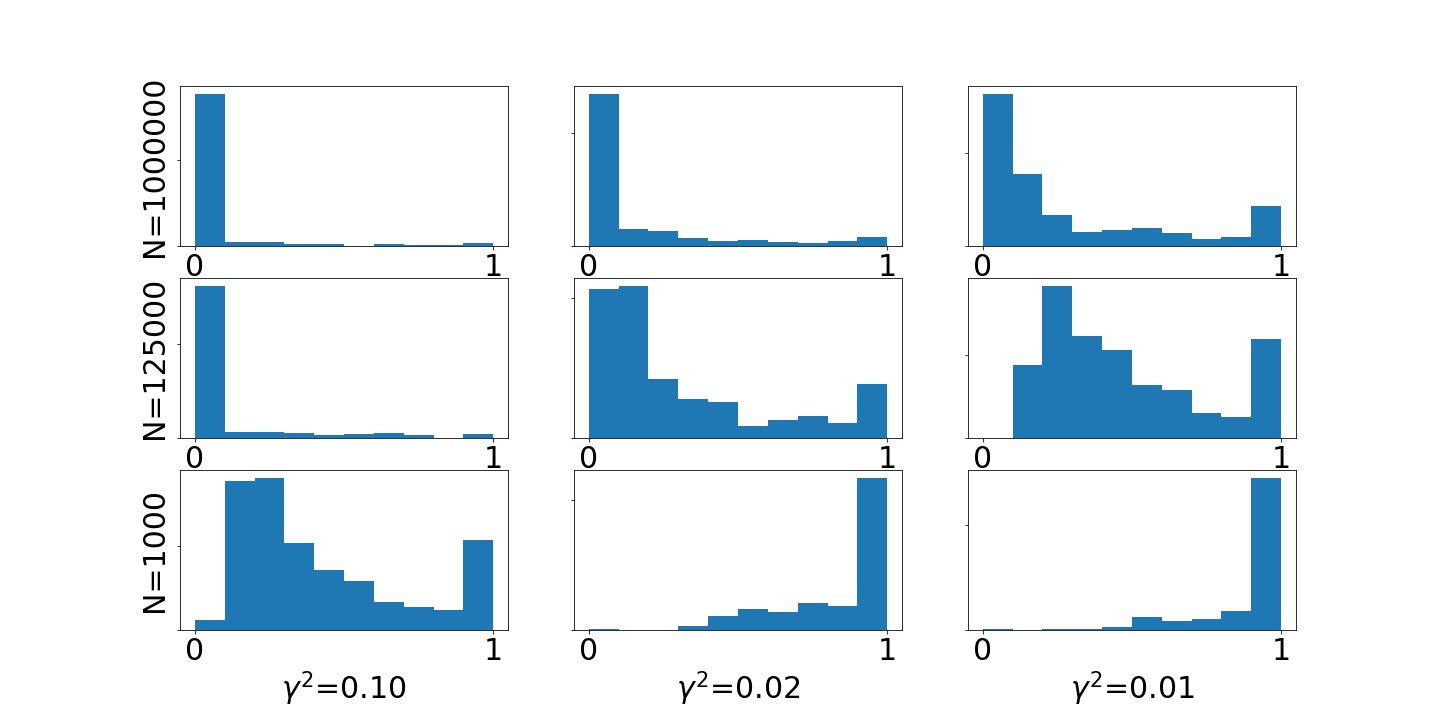}
        \subcaption{$N = \gamma^{-6}$.}
        \label{fig:prob1_6_1}
    \end{minipage}
    \caption{Noise scaling with $d = k =  5.$ \label{fig:Noise Scaling d=5}}
\end{figure}

We can see from Figure \ref{fig:Noise Scaling d=5}.a that $N=\gamma^{-4}$ is not a fast enough growth of $N$ to avoid weight collapse: the histograms skew to the right along the bottom-left to top-right diagonal, suggesting that weight collapse (i.e. one weight dominating the rest, and therefore the variance of the weights being large) is bound to occur in the joint limit $N\to \infty,$ $\gamma\to 0$ with $N=\gamma^{-4}$. In contrast, the histograms in Figure \ref{fig:Noise Scaling d=5}.b skew to the left along the same diagonal, suggesting that  the probability of weight collapse is significantly reduced if $N =\gamma^{-6}$. We observe a similar behavior with other choices of dimension $d$ by conducting experiments with sample sizes $N=\gamma^{-d+1}$ and $N=\gamma^{-d-1}$, and we include the histograms with $d=k=4$ in Appendix \ref{AppendixC}. Our next result shows that these empirical findings are in agreement with the scaling of the $\chi^2$-divergence between target and proposal in the small noise limit. 
\begin{proposition}\label{prop:noise_scaling}
    Consider the inverse problem setting
    \begin{equation*}
    y=\A u+\eta, \quad \quad \eta  =\Nc(0,\gamma^2\Gamma)
    , \quad \quad u\sim  \pi = \Nc(0,\Sigma).
    \end{equation*}
    Let $\mu_\gamma$ denote the posterior and let $\rho_\gamma = \dchi( \mu_\gamma\| \pi) +1.$  Then, for almost every $y,$  
	\[\rho_\gamma\sim\bigO(\gamma^{-k})\]
	 in the small noise limit $\gamma\to 0.$
\end{proposition}

\begin{proof}
	Let $K_\gamma= \Sigma \A'(\A\Sigma \A'+\gamma^2\Gamma)^{-1}$ denote the Kalman gain. We observe that $K_\gamma\to \Sigma A' (A \Sigma A')^{-1}$ under our standing assumption that $A$ is full rank. 
    Let $U'\Xi V$ be the singular value decompostion of $\Gamma^{-\frac{1}{2}}A\Sigma^{\frac{1}{2}}$  and $\{\xi_i\}_{i=1}^k$ be the singular values.
    Then we have
        \begin{align*}
      K_\gamma \A
        &\sim \Sigma^{\frac{1}{2}}A'\Gamma^{-\frac{1}{2}}(\Gamma^{-\frac{1}{2}}A\Sigma A'\Gamma^{-\frac{1}{2}}+\gamma^2I)^{-1}\Gamma^{-\frac{1}{2}}A\Sigma^{\frac{1}{2}}\\
        &= V'\Xi' U(U'\Xi VV'\Xi' U+\gamma^2 I)^{-1}U'\Xi V\\
        &\sim \Xi'(\Xi\Xi'+\gamma^2I)^{-1}\Xi,
        \end{align*}
     where here ``$\sim$'' denotes matrix similarity.
        It follows that $I + K_\gamma A$ converges to a finite limit, and so does the exponent $y' K_\gamma'\Sigma^{-1}(I+ K_\gamma A)^{-1}K_\gamma y$ in Proposition \ref{prop:inverse_problem_chi2}.
On the other hand, 
        \begin{align*}
        (|I + K_\gamma \A||I- K_\gamma \A|)^{-\frac{1}{2}} 
        = \Bigl(\prod_{i=1}^{k}\frac{\gamma^2}{\xi_i^2+\gamma^2}\Bigr)^{-\frac{1}{2}}
        \sim\bigO(\gamma^{-k})
        \end{align*}
    as $\gamma\to 0$. The conclusion follows.
\end{proof}

%

\subsection{Importance Sampling and Prior Scaling}\label{ssec:ippriorscaling}
Here we illustrate the behavior of importance sampling in the limit of large prior covariance. 
We start again with a motivating numerical example, similar to the one reported in  Figure \ref{fig:Noise Scaling d=5}. The behavior is analogous to the small noise regime, which is expected since the \emph{ratio} of prior and noise covariances determines the closeness between target and proposal. Figure \ref{fig:Prior Scaling d=5}  shows that when $d=k=5$ weight collapse is observed frequently when the sample size $N$ grows as $\sigma^{4},$ but not so often with sample size $N=\sigma^{6}$. Similar histograms with $d=k=4$ are included in Appendix \ref{AppendixC}. These empirical results are in agreement with the theoretical growth rate of the $\chi^2$-divergence between target and proposal in the limit of large prior covariance, as we prove next.
\FloatBarrier
	
\begin{figure}[!htb] 
    \begin{minipage}{.5\textwidth}
       
        \includegraphics[width=1\linewidth, height=0.255\textheight]{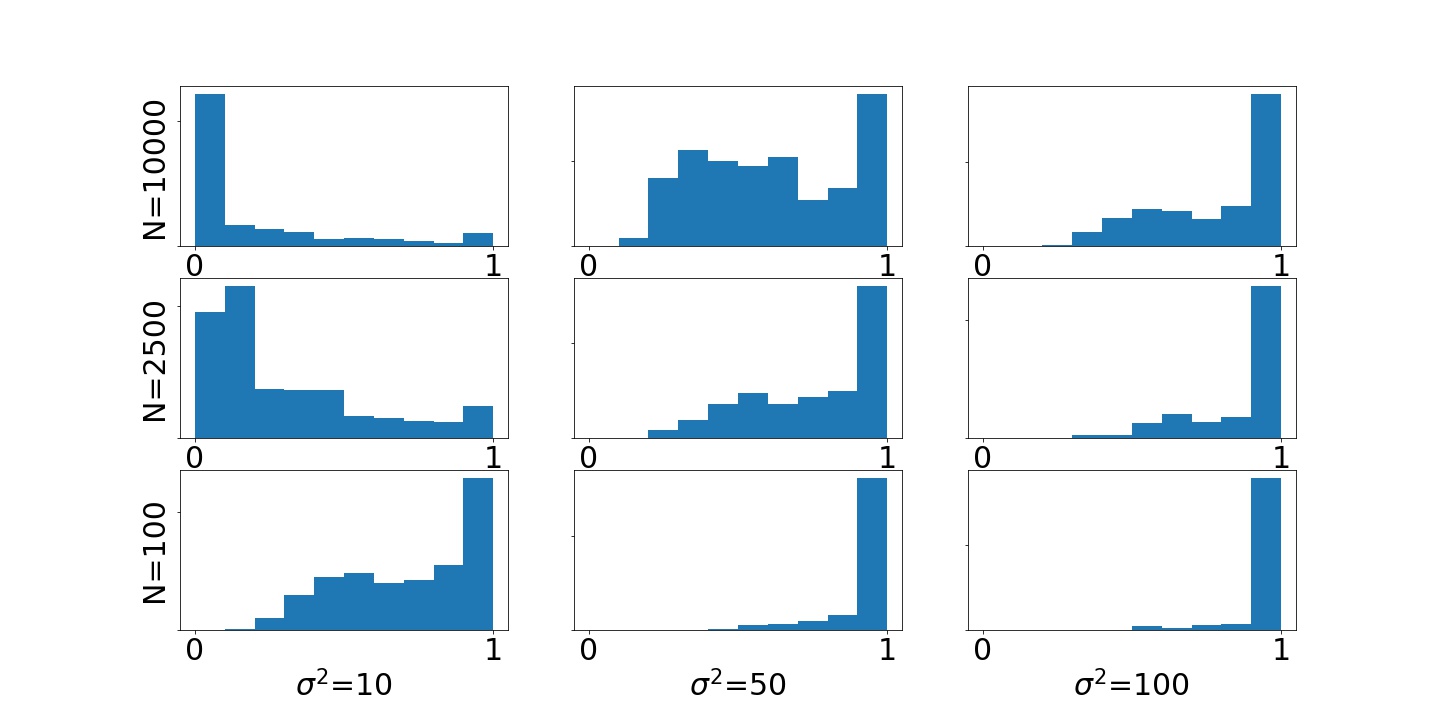}
        \subcaption{$N=\sigma^{4}$.}
        \label{fig:prob1_6_2}
    \end{minipage}%
    \begin{minipage}{0.5\textwidth}
       
        \includegraphics[width=1\linewidth, height=0.255\textheight]{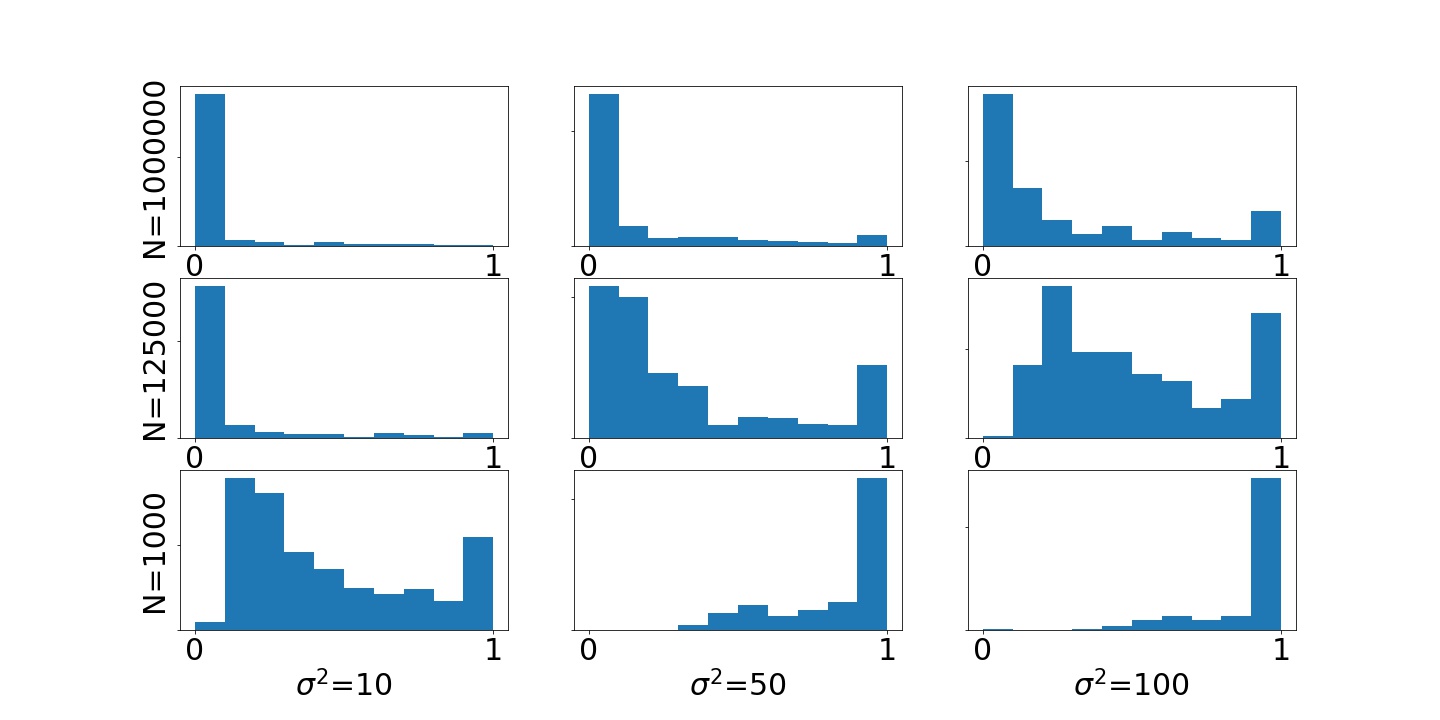}
        \subcaption{$N = \sigma^{6}$.}
        \label{fig:prob1_6_1}
    \end{minipage}
    \caption{Prior scaling $d = k =  5.$ \label{fig:Prior Scaling d=5}}
\end{figure}

\FloatBarrier

\begin{proposition}\label{prop:prior_scaling}
    Consider the inverse problem setting
    \begin{equation*}
    y=\A u+\eta, \quad \quad \eta \sim \Nc(0,\Gamma)
    , \quad \quad u\sim \pi_\sigma = \Nc(0,\sigma^2\Sigma).
    \end{equation*}
    Let $\mu_\sigma$ denote the posterior and $\rho_\sigma = \dchi(\mu_\sigma \| \pi_\sigma) + 1.$ Then, for almost every $y,$
	\[  \rho_\sigma\sim\bigO(\sigma^d)\]
	in the large prior limit $\sigma \to \infty.$
\end{proposition}

\begin{proof}
Let $\Sigma_\sigma=\sigma^2\Sigma$, let $K_\sigma=\Sigma_\sigma \A'(\A \Sigma_\sigma  \A'+\Gamma)^{-1}$ be the Kalman gain. Observing that $K_\sigma=K_{\gamma=\frac{1}{\sigma}}$, we apply Proposition \ref{prop:noise_scaling} and deduce that when $\sigma\to\infty$: 
    \begin{enumerate}
        \item $K_\sigma\to \Sigma \A'(\A\Sigma \A'+\gamma^2\Gamma)^{-1}$;
        \item $I+K_\sigma A$ has a well-defined and invertible limit;
        \item $|I-K_\sigma A|^{-\frac{1}{2}}\sim\bigO(\sigma^k)$.
    \end{enumerate}
On the other hand, we notice that the quadratic term
    \[K_\sigma'\Sigma_\sigma^{-1}(I+K_\sigma A)^{-1}K_\sigma
    =\sigma^{-2}K_\sigma'\Sigma(I+K_\sigma A)^{-1}K_\sigma\]
    vanishes in limit.  The conclusion follows by Proposition \ref{prop:inverse_problem_chi2}.
\end{proof}

\subsection{Importance Sampling in High Dimension}\label{ssec:ipdimension}
In this subsection we study importance sampling in high dimensional limits. To that end, we let
  $\{a_i\}_{i =1}^\infty,$ $\{\gamma_i^2\}_{i=1}^\infty$ and  $\{\sigma_i^2\}_{i=1}^\infty$ be infinite sequences and we define, for any $d\ge 1,$
	\begin{align*}
	A_{1:d} &:= \text{diag} \Bigl\{ a_1, \ldots, a_d\Bigr\} \in \R^{d\times d}, \\
	\Gamma_{1:d} &:= \text{diag} \Bigl\{ \gamma_1^2, \ldots, \gamma_d^2\Bigr\} \in \R^{d\times d}, \\
	\Sigma_{1:d} &:= \text{diag} \Bigl\{ \sigma_1^2, \ldots, \sigma_d^2\Bigr\}  \in \R^{d\times d}.
	\end{align*}
	We then consider the inverse problem of reconstructing $u \in \R^d$ from data $y \in \R^d$ under the setting
	\begin{equation}\label{eq:ddimip}
	y= A_{1:d} u + \eta, \quad \quad \eta \sim \Nc(0,\Gamma_{1:d}), \quad \quad  u \sim \pi_{1:d} = \Nc(0,\Sigma_{1:d}).
	\end{equation}
 We denote the corresponding posterior distribution by $\mu_{1:d},$ which  is Gaussian with a diagonal covariance. Given observation $y$, we may find the posterior distribution $\mu_i$ of $u_i$ by solving the one dimensional linear-Gaussian inverse problem
 \begin{equation}\label{eq:1dipsetting}
 y_i = a_i u_i+ \eta_i, \quad \quad \eta_i \sim \Nc(0,\gamma_i^2), \quad \quad 1\le i \le d,
 \end{equation}
with prior $\pi_i= \Nc(0, \sigma_i^2).$ In this way we have defined, for each $d \in \N \cup \{\infty\},$ an inverse problem with prior and posterior
\begin{equation}\label{eq:product structure}
\pi_{1:d} = \prod_{i=1}^d \pi_i, \quad \quad  \mu_{1:d} = \prod_{i=1}^d \mu_i.
\end{equation}
In Subsection \ref{ssec:onedimensional} we include an explicit calculation in the one dimensional inverse setting \eqref{eq:1dipsetting}, which will be used in Subsection \ref{ssec:dalargedimension} to establish the rate of growth of $\rho_d = \dchi(\mu_{1:d} \| \pi_{1:d})$ and thereby how the sample size needs to be scaled along the high dimensional limit $d\to\infty$ to maintain the same accuracy. Finally, in Subsection \ref{ssec:infinited} we establish from first principles and our simple one dimensional calculation the equivalence between $(i)$ certain notion of dimension being finite; $(ii)$ $\rho_\infty< \infty;$ and $(iii)$ absolute continuity of $\mu_{1:\infty}$ with respect to $\pi_{1:\infty}.$

\subsubsection{One Dimensional Setting}\label{ssec:onedimensional}
  Let $a \in \R$ be given and consider the one dimensional inverse problem of reconstructing $u\in \R$ from data $y \in \R$, under the setting
  \begin{equation}\label{eq:1dproblem}
  y = au + \eta, \quad \quad \eta \sim \Nc(0, \gamma^2), \quad \quad  u \sim \pi = \Nc(0,\sigma^2).
  \end{equation}
   By defining
     $$g(u):= \exp\Bigl( -\frac{a^2}{2\gamma^2}u^2 + \frac{ay}{\gamma^2} u  \Bigr) ,$$
    we can write the posterior density $\mu(du)$ as  $\mu(du) \propto g(u) \pi(du).$ The next result gives a simplified closed formula for $\rho = \dchi(\mu \| \pi) + 1.$ In addition, it gives a closed formula for the Hellinger integral 
    $$\mathcal{H}(\mu, \pi) :=\frac{ \pi \bigl(g^{\frac12} \bigr) }{ \pi(g)^{\frac12}},$$
  which will facilitate the study of the case $d = \infty$ in Subsection \ref{ssec:infinited}. 
\begin{lemma}\label{lem: oned}
Consider the inverse problem in \eqref{eq:1dproblem}. Let $\lambda := a^2 \sigma^2/\gamma^2$ and $z^2:= \frac{y^2}{a^2 \sigma^2 + \gamma^2}.$ Then,  for any $\ell >0,$
\begin{align}
    \frac{\pi(g^\ell)}{\pi(g)^\ell} 
    &=\frac{(\lambda+1)^{\frac{\ell}{2}}}{\sqrt{\ell\lambda+1}}\exp\Bigl(\frac{(\ell^2-\ell)\lambda}{2(\ell\lambda+1)}z^2\Bigr).\label{eq:g k moment}
    \end{align}
In particular, 
\begin{align}
    \rho
    &=\frac{\lambda + 1}{ \sqrt{2\lambda + 1} }\exp \Bigl( \frac{\lambda}{2\lambda + 1} z^2 \Bigr),\label{eq:chi2oned}\\
    \mathcal{H}(\mu, \pi)  
    &=\sqrt{\frac{2\sqrt{\lambda+1}}{\lambda+2}} \exp\Bigl(-\frac{\lambda z^2}{4(\lambda+2)} \Bigr).\label{eq:Hellingeroned}
\end{align}
\end{lemma}
\begin{proof}
A direct calculation shows that 
$$  \post(g) = \frac{1}{\sqrt{\lambda + 1}} \exp\Bigl( \frac12 \frac{\lambda y^2}{a^2 \sigma^2 + \gamma^2}  \Bigr).$$
The same calculation, but replacing $\gamma^2$ by $\gamma^2/\ell$ and $\lambda$ by $\ell\lambda$, gives similar expressions for $\pi(g^\ell)$, which leads to (\ref{eq:g k moment}). The other two equations follow by setting $\ell$ to be $2$ and $\frac{1}{2}$. 
\end{proof}

Lemma \ref{lem: oned} will be used in the two following subsections to study high dimensional limits. Here we show how this lemma also allows us to verify directly that the assumption $\mathcal{V}<1$ in Theorem \ref{thm:nece suff summary} holds in small noise and large prior limits.

\begin{example}\label{example: 1-d gaussian, verify assumption of sup var/rho2}
 Consider a sequence of inverse problems of the form \eqref{eq:1dproblem} with $\lambda = a^2 \sigma^2/\gamma^2$ approaching infinity. Let $\{(\mu_\lambda,\pi_\lambda)\}_{\lambda>0}$ be the corresponding family of posteriors and priors and let $\g_\lambda$ be the normalized density.  Lemma \ref{lem: oned} implies that 
    \begin{align*}
   \frac{ \pi_\lambda(\g_\lambda^3)}{\pi_\lambda(\g_\lambda^2)^2}
    =&\frac{2\lambda+1}{\sqrt{(3\lambda+1)(\lambda+1)}}\exp\Bigl(\frac{\lambda}{(2\lambda+1)(3\lambda+1)}z^2\Bigr)
    \to\frac{2}{\sqrt{3}}
    <2,
    \end{align*}
as $\lambda \to \infty$. This implies that, for $\lambda$ sufficiently large,
    \[\frac{\V[\g_\lambda(U_\lambda)]}{\rho_\lambda^2}=\frac{\pi_\lambda(\g_\lambda^3)}{\pi_\lambda(\g_\lambda^2)^2}-1<1.  \]
\qed
\end{example}

\subsubsection{ Large Dimensional Limit}
Now we investigate the behavior of importance sampling in the limit of large dimension, in the inverse problem setting \eqref{eq:ddimip}. We start with an example similar to the ones in Figure $\ref{fig:Noise Scaling d=5}$ and Figure $\ref{fig:Prior Scaling d=5}$. Figure \ref{fig:Dimensional Scaling lambda=1.3} shows that for $\lambda=1.3$ fixed, weight collapse happens frequently when the sample size $N$ grows polynomially as $d^2$, but not so often if  $N$ grows at rate
$\bigO\left(\prod_{i=1}^d\left(\frac{\lambda+1}{\sqrt{2\lambda+1}}e^{\frac{\lambda z_i^2}{2\lambda +1}}\right)\right)$. 
Similar histograms for $\lambda=2.4$ are included in Appendix \ref{AppendixC}. These empirical results are in agreement with the growth rate of $\rho_d$ in the large $d$ limit.

\FloatBarrier
    
\begin{figure}[!htb] 
    \begin{minipage}{.5\textwidth}
       
        \includegraphics[width=1\linewidth, height=0.255\textheight]{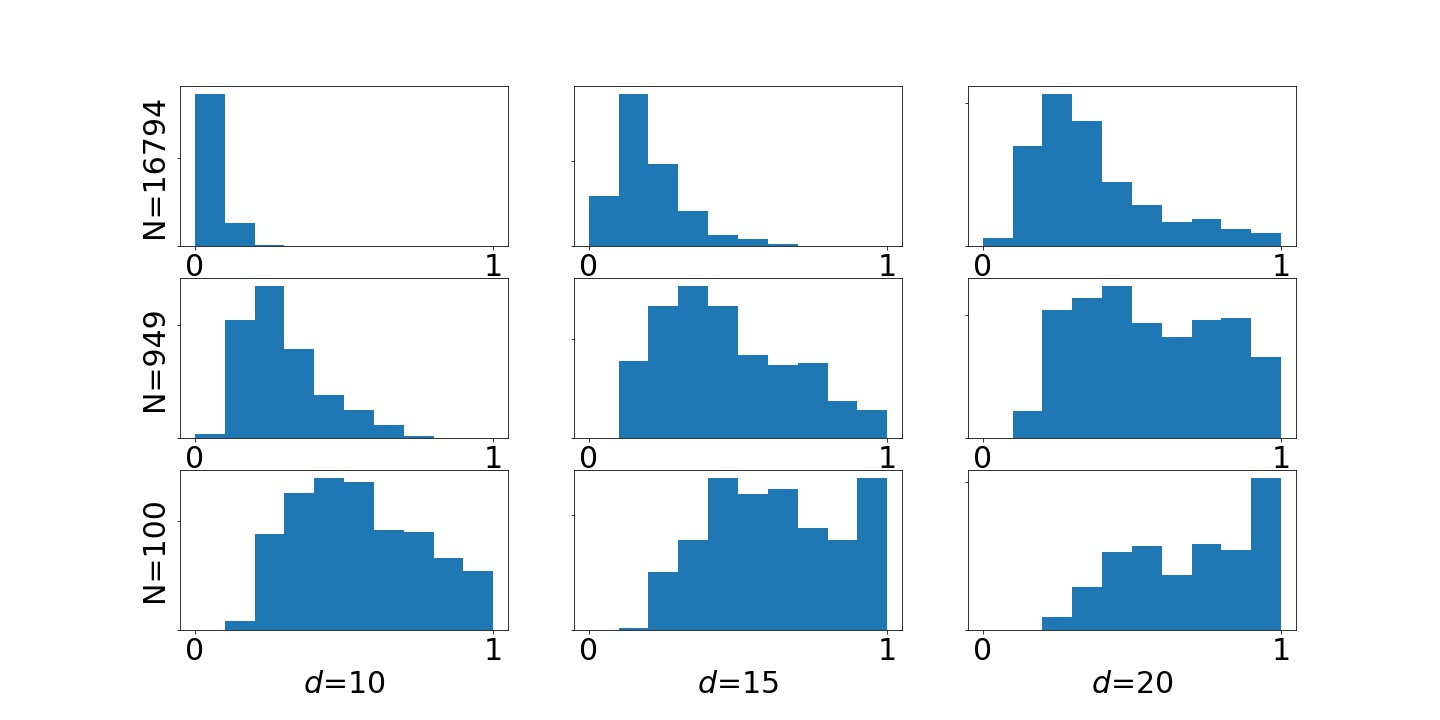}
        \subcaption{$N=\bigO\left(\prod_{i=1}^d\left(\frac{\lambda+1}{\sqrt{2\lambda+1}}e^{\frac{\lambda z_i^2}{2\lambda +1}}\right)\right)$.}
        \label{fig:dim no collapse}
    \end{minipage}%
    \begin{minipage}{0.5\textwidth}
       
        \includegraphics[width=1\linewidth, height=0.255\textheight]{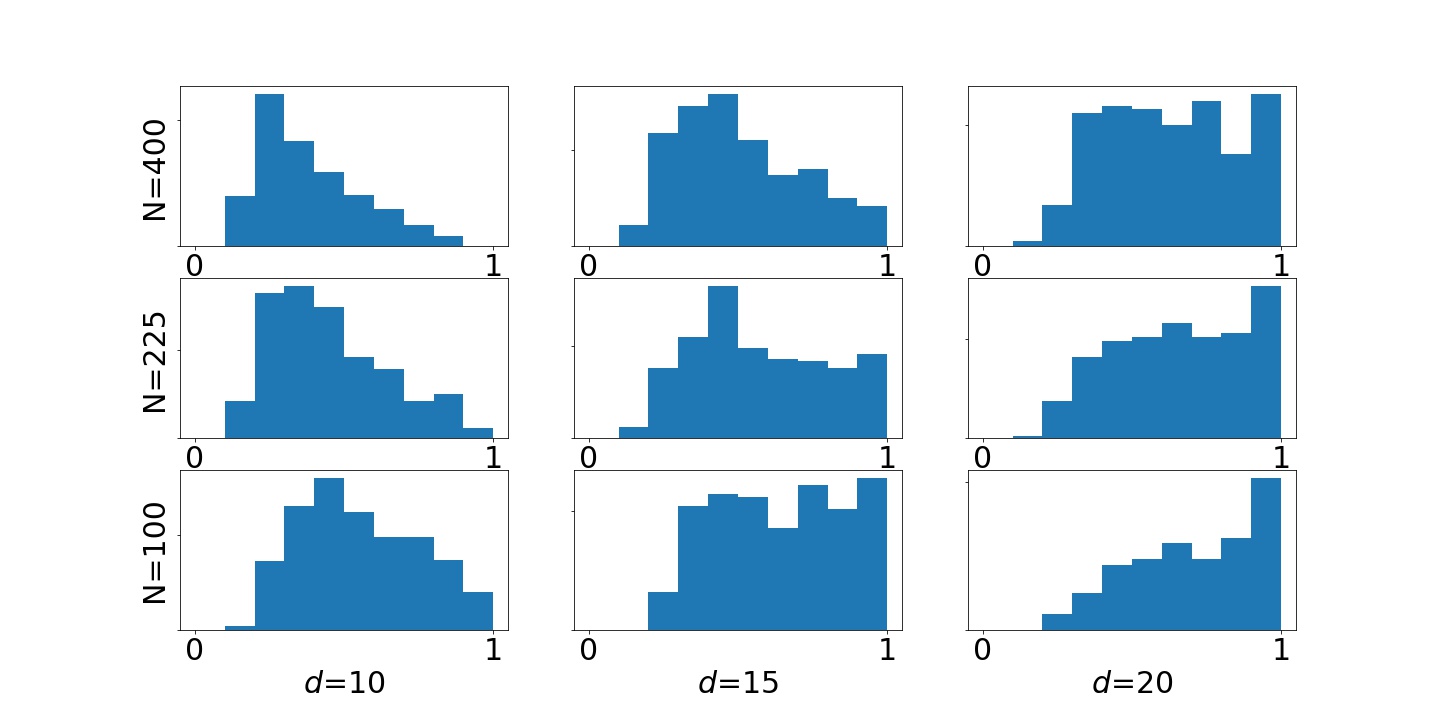}
        \subcaption{$N=d^2$.}
        \label{fig:dim collapse}
    \end{minipage}
    \caption{Dimensional scaling $\lambda = 1.3.$ \label{fig:Dimensional Scaling lambda=1.3}}
\end{figure}

\FloatBarrier

\begin{proposition}\label{thm: chi2}
For any $d \in \N \cup \{\infty \},$
\begin{align*}
\rho_d  &= \prod_{i=1}^d\left(\frac{\lambda_i+1}{\sqrt{2\lambda_i+1}}e^{\frac{\lambda_iz_i^2}{2\lambda_i+1}}\right),  \\
\E_{z_{1:d}}\left[\rho_d\right] &= \prod_{i = 1}^d (\lambda_i+1).
\end{align*}
\end{proposition}
\begin{proof}
The formula for $\rho_d$ is a direct consequence of Equation \eqref{eq:chi2oned} and the product structure. Similarly, we have
    \begin{align*}
        \E_{z_i}\left[\frac{\lambda_i+1}{\sqrt{2\lambda_i+1}}e^{\frac{\lambda_iz_i^2}{2\lambda_i+1}}\right]
        &=\frac{\lambda_i+1}{\sqrt{2\lambda_i+1}}\int_{\R} \frac{1}{\sqrt{2\pi}}e^{-\frac{z_i^2}{2}+\frac{\lambda_iz_i^2}{2\lambda_i+1}}dz_i\\
        &=\frac{\lambda_i+1}{\sqrt{2\lambda_i+1}}\int_{\R} \frac{1}{\sqrt{2\pi}}e^{-\frac{z_i^2}{2(2\lambda_i+1)}}dz_i\\
        &=\lambda_i+1.
    \end{align*}
\end{proof}

Proposition  \ref{thm: chi2} implies that, for $d \in \N \cup \{\infty \},$
\begin{align*}
\sup_{|\phi|_\infty\leq 1}\Expect \left[ \bigl(\muis^N_{1:d}(\phi)-\mu_{1:d}(\phi)\bigr)^2 \right]  & \le 4  \prod_{i=1}^d\left(\frac{\lambda_i+1}{\sqrt{2\lambda_i+1}}e^{\frac{\lambda_iz_i^2}{2\lambda_i+1}}\right), \\
\E \Biggl[ \sup_{|\phi|_\infty\leq 1}\Expect \left[ \bigl(\muis^N_{1:d} (\phi)-\mu_{1:d}(\phi)\bigr)^2 \right]  \Biggr]  & \le 4  \prod_{i = 1}^d (\lambda_i + 1). 
\end{align*}
Note that the outer expected value in the latter equation averages over the data, while the inner one averages over sampling from the prior $\pi_{1:d}$. This suggests that 
$$\log \E \Biggl[ \sup_{|\phi|_\infty\leq 1}\Expect \left[ \bigl(\muis^N_{1:d} (\phi)-\mu_{1:d}(\phi)\bigr)^2 \right] \Biggr] \lesssim \sum_{i=1}^d \lambda_i.  $$
The quantity $\tau := \sum_{i=1}^d\lambda_i$ had been used as an \emph{intrinsic dimension} of the inverse problem \eqref{eq:ddimip}. This simple heuristic together with Theorem \ref{thm:nece suff summary} suggest that increasing $N$ \emph{exponentially} with $\tau$ is both necessary and sufficient to maintain accurate importance sampling along the high dimensional limit $ d\to \infty.$ In particular, if all coordinates of the problem play the same role, this implies that $N$ needs to grow exponentially with $d$, a manifestation of the curse of dimension of importance sampling \cite{agapiou2017importance,BBL08,BLB08}.  

\subsubsection{Infinite Dimensional Singularity}\label{ssec:infinited}
	Finally, we investigate the case $d = \infty.$ Our goal in this subsection is to establish a connection between the effective dimension, the quantity $\rho_\infty$, and absolute continuity. The main result, Corollary \ref{cor: dimensionality 3 equivalence}, had been proved in more generality in \cite{agapiou2017importance}. However, our proof and presentation here requires minimal technical background and is based on the explicit calculations obtained in the previous subsections and in the following lemma. 
	\begin{lemma}\label{lem: hellinger integral abs cts}
	It holds that $\mu_{1:\infty}$ is absolutely continuous with respect to $\pi_{1:\infty}$ if and only if 
	\begin{equation}
\mathcal{H}(\mu_{1:\infty}, \pi_{1:\infty})  =	\prod_{i=1}^\infty \frac{\pi_i \bigl(g_i ^{\frac12} \bigr) }{ \pi_i(g_i)^{\frac12}}  >0,
	\end{equation}
	where $g_i$ is an unnormalized density between $\mu_i$ and $\pi_i.$
	Moreover, we have the following explicit characterizations of the Hellinger integral $\mathcal{H}(\mu_{1:\infty}, \pi_{1:\infty})$ and its average with respect to data realizations, 
	\begin{align*}
		\mathcal{H}(\mu_{1:\infty}, \pi_{1:\infty})  &=  \prod_{i=1}^\infty \left( \sqrt{\frac{2\sqrt{\lambda_i+1}}{\lambda_i+2}} e^{-\frac{\lambda_i z_i^2}{4(\lambda_i+2)}}  \right),   \\
        \E_{z_{1:\infty}}\left[\mathcal{H}(\mu_{1:\infty}, \pi_{1:\infty})\right] &= \prodi\frac{2(\lambda_i+1)^{\frac{1}{4}}}{\sqrt{3\lambda_i+4}}.
    \end{align*}
	\end{lemma}

\begin{proof}
The formula for the Hellinger integral is a direct consequence of Equation \eqref{eq:Hellingeroned} and the product structure.
 On the other hand, 
    \begin{align*}
        \E_{z_i}\left[\sqrt{\frac{2\sqrt{\lambda_i+1}}{\lambda_i+2}} e^{-\frac{\lambda_i z_i^2}{4(\lambda_i+2)}}\right]
        &= \frac{\sqrt{2}(\lambda_i+1)^{\frac{1}{4}}}{\sqrt{\lambda_i+2}}\int_{\R}\frac{1}{\sqrt{2\pi}}e^{-\frac{\lambda_iz_i^2}{4(\lambda_i+2)}-\frac{z_i^2}{2}}dz_i\\
        &=\frac{2(\lambda_i+1)^{\frac{1}{4}}}{\sqrt{3\lambda_i+4}}.
    \end{align*}
    The proof of the equivalence between finite Hellinger integral and absolute continuity is given in Appendix \ref{aappendixB}.
\end{proof}

\begin{corollary}\label{cor: dimensionality 3 equivalence}
The following statements are equivalent:
	\begin{enumerate}[(i)]
		\item $\tau=\sumi \lambda_i <\infty$;
		\item $\rho_\infty<\infty$ for almost every $y$;
		\item $\mu_{1:\infty}\ll\pi_{1:\infty}$ for almost every $y$.
	\end{enumerate}
\end{corollary}
\begin{proof}
    Observe that $\lambda_i\to 0$ is a direct consequence of all three statements, so we will assume $\lambda_i\to 0$ from now on.\\
    $(i)\logeq(ii):$
    By Proposition \ref{thm: chi2}, 
        \begin{align*}
            \log\Bigl(\E_{z_{1:\infty}}\left[\rho_\infty\right]\Bigr)=\sumi\log(1+\lambda_i)=\bigO(\sumi\lambda_i),
        \end{align*}
    since $\log(1+\lambda_i)\approx\lambda_i$ for large $i$.\\
    $(i)\logeq(iii):$
    Similarly, we have
        \begin{align*}
            \log\Bigl(\E_{z_{1:\infty}}\left[\mathcal{H}(\mu_{1:\infty}, \pi_{1:\infty})\right]\Bigr) 
            &= -\frac{1}{4}\sumi\log\frac{(3\lambda_i+4)^2}{16(\lambda_i+1)}\\
            &= -\frac{1}{4}\sumi\log\left(1+\frac{9\lambda_i^2+8\lambda_i}{16\lambda_i+16}\right)\\
            &= -\frac{1}{4}\bigO(\sumi\lambda_i).
        \end{align*}
    The conclusion follows from Lemma \ref{lem: hellinger integral abs cts}.
\end{proof}


\section{Importance Sampling for Data Assimilation}\label{sec:4}
In this section, we study the use of importance sampling in a particle filtering setting. Following  \cite{BBL08,BLB08,snyder2008obstacles} we focus on one filtering step. Our goal is to provide a new and concrete comparison of two proposals, referred to as \emph{standard} and \emph{optimal} in the literature \cite{agapiou2017importance}. In Subsection \ref{ssec:onestepsetting} we introduce the setting and both proposals, and show that the $\chi^2$-divergence between target and standard proposal is larger than the $\chi^2$-divergence between target and optimal proposal. Subsections \ref{ssec:dasmallnoise} and \ref{ssec:dalargedimension} identify small noise and large dimensional limiting regimes where the sample size for the standard proposal needs to grow unboundedly to maintain the same level of accuracy, but the required sample size for the optimal proposal remains bounded. 
\subsection{One-step Filtering Setting}\label{ssec:onestepsetting}
Let $M$ and $H$ be given matrices. We consider the one-step filtering problem of recovering $v_0,v_1$ from $y$, under the following setting
    \begin{align}
   v_1 &=Mv_0+\xi,\label{eq: filt step 1} \quad \quad
    v_0\sim\Nc(0,P), \quad
    \xi\sim\Nc(0,Q),\\
    y &=H v_1+\zeta,\label{eq: filt step 2} \quad \quad
    \zeta \sim\Nc(0,R).
    \end{align}
Similar to the setting in Subsection \ref{ssec:ipsetting}, we assume that $P,Q,R$ are symmetric positive definite and that $M$ and $H$ are full rank. From a Bayesian point of view, we would like to sample from the target distribution $\Prob_{v_0,v_1|y}$. To achieve this, we can either use $\pis=\Prob_{v_1|v_0}\Prob_{v_0}$ or $\pio=\Prob_{v_1|v_0,y}\Prob_{v_0}$ as the proposal distribution.

The standard proposal $\pis$ is the prior distribution of $(v_0,v_1)$ determined by the prior $v_0 \sim \Nc(0,P)$ and the \emph{signal dynamics} encoded in  Equation (\ref{eq: filt step 1}). Then assimilating the observation $y$ leads to an inverse problem \cite{agapiou2017importance,sanzstuarttaeb} with design matrix, noise covariance, and prior covariance given by     
    \begin{align}\label{eq:standardproposalsdef}
    \begin{split}
    \As &:=H, \\ 
    \Gammas &:=R, \\
    \Sigmas &:=MPM'+Q. 
    \end{split}
    \end{align}
    We denote $\pis = \Nc(0, \Sigmas)$ the prior distribution and by $\mus$ the corresponding posterior distribution. 

The optimal proposal $\pio$ samples from $v_0$ and the conditional kernel $v_1|v_0,y.$ Then assimilating $y$ leads to the inverse problem \cite{agapiou2017importance,sanzstuarttaeb}
    \[y=HMv_0+H\xi+\zeta,\]
where the design matrix, noise covariance and prior covariance are given by     
    \begin{align}\label{eq:optimalproposalsdef}
    \begin{split}
    \Ao &:=HM, \\
    \Gammao &:=HQH'+R, \\ 
    \Sigmao &:=P. 
    \end{split}
    \end{align}
We denote $\pio = \Nc(0, \Sigmao)$ the prior distribution and $\mus$ the corresponding posterior distribution.

\subsection{\texorpdfstring{$\chi^2$}{chi2}
-divergence Comparison between Standard and Optimal Proposal}
Here we show that $$\rhos:= \dchi(\mus \| \pis)  + 1 > \dchi(\muo \| \pio)  + 1 =: \rhoo.$$ The proof is a direct calculation using the explicit formula in Proposition \ref{prop:inverse_problem_chi2}.
We introduce, as in Section \ref{sec:3}, standard Kalman notation
  \begin{align*}
  \Ks := \Sigmas \As' \Ss^{-1}, \quad \quad \quad \Ss := \As \Sigmas \As' + \Gammas, \\
    \Ko := \Sigmao \Ao' \So^{-1}, \quad \quad \quad \So := \Ao \Sigmao \Ao' + \Gammao.
\end{align*}   
It follows from the definitions in \eqref{eq:standardproposalsdef} and \eqref{eq:optimalproposalsdef} that
\begin{align*}
\Ss &= H(MPM' + Q)H + R \\
 &= HM  P M' H + HQH' + R \\
 & = \So.
\end{align*}
Since  $\Ss = \So$ we drop the subscripts in what follows, and denote both simply by $S.$


    \begin{theorem}\label{prop: part filt chi2}
    Consider the one-step filtering setting in Equations \eqref{eq: filt step 1} and \eqref{eq: filt step 2}.  If $M$ and $H$ are full rank and $P,Q,R$ are symmetric positive definite, then, for almost every $y,$
        \[\rhos>\rhoo. \]
    \end{theorem}
    \begin{proof} 
    By Proposition \ref{prop:inverse_problem_chi2} we have
        \begin{align*}
        \rhos
        &=
        (|I-\Ks\As||I+\Ks \As|)^{-\frac{1}{2}}
        \exp \Bigl( y'\Ks' [(I+ \Ks \As)\Sigmas]^{-1} \Ks \, y \Bigr),\\          
     \rhoo & =   (|I-\Ko\Ao||I+ \Ko\Ao|)^{-\frac{1}{2}}
        \exp\Bigl( y'\Ko' [ (I+ \Ko \Ao) \Sigmas]^{-1} \Ko \, y \Bigr).
        \end{align*}
     Therefore, it suffices to prove the following two inequalities:
        \begin{align}
            |I- \Ks \As||I+ \Ks \As|
            &< |I- \Ko \Ao||I+ \Ko \Ao|,\label{ineq: part filt 1}\\
           \Ks'  [(I+ \Ks \As)\Sigmas]^{-1} \Ks
            &\prec \Ko'  [ (I+ \Ko \Ao) \Sigmas]^{-1}\Ko.\label{ineq: part filt 2}
        \end{align}
    We start with inequality  \eqref{ineq: part filt 2}. Note that 
        \begin{align*}
        (I+ \Ks \As)\Sigmas &= \Sigmas + \Sigmas \As' S^{-1} \As \Sigmas, \\
          (I+ \Ko \Ao)\Sigmao &= \Sigmao + \Sigmao \Ao' S^{-1} \Ao \Sigmao.
        \end{align*}
    Using the definitions in \eqref{eq:standardproposalsdef} and \eqref{eq:optimalproposalsdef} it follows that
        \begin{align*}
         \Ks' \Sigmas^{-1} (I+ \Ks \As)^{-1} \Ks & = H \Bigl\{ (MPM' + Q)^{-1} + H' S H \Bigr\}^{-1} H' \\
        &\prec H \Bigl\{ (MPM')^{-1} + H' S H \Bigr\}^{-1} H' \\
         &= \Ko' \Sigmao^{-1} (I+ \Ko \Ao)^{-1} \Ko.
        \end{align*}    
    For inequality \eqref{ineq: part filt 1}, we notice that
        \begin{align*}
        &\Ks \As = (MPM'+Q)H' S^{-1} H
        =M\tilde{P}M'H' S^{-1} H
        \sim (H' S^{-1}  H)^{\frac{1}{2}}M\tilde{P}M'(H' S^{-1}  H)^{\frac{1}{2}},\\
        &\Ko \Ao = PM'H' S^{-1}  HM
        \sim (H'S^{-1}H)^{\frac{1}{2}}MPM'(H' S^{-1}  H)^{\frac{1}{2}},
        \end{align*}
    where $\tilde{P}:=P+M^{\dagger}QM^{'\dagger}.$ Therefore
   \begin{equation*}
   \Ko \Ao \prec \Ks \As
   \end{equation*}
which, together with $\Ks \As \prec I,$ implies that
        \begin{align*}
        |I- \Ks \As||I+ \Ks \As|
        -|I-\Ko \Ao||I+\Ko \Ao|
        =&|I-(\Ks \As)^2|-|I-(\Ko \Ao)^2|>0,
        \end{align*}
        as desired. 
    \end{proof}

    \subsection{Standard and Optimal Proposal in Small Noise Regime}\label{ssec:dasmallnoise}
    It is possible that along a certain limiting regime, $\rho$  diverges for the standard proposal, but not for the optimal proposal. The proposition below gives an explicit example of this scenario. Precisely, consider the following one-step filtering setting
    \begin{align*}
   v_1 &=Mv_0+\xi, \quad \quad
    v_0\sim\Nc(0,P), \quad
    \xi\sim\Nc(0,Q),\\
    y &=H v_1+\zeta, \quad \quad
    \zeta \sim\Nc(0,r^2R),
    \end{align*}
    where $r\to 0$. Let $\muo^{(r)},\mus^{(r)}$ be the optimal/standard targets and $\pio^{(r)},\pis^{(r)}$ be the  optimal/standard proposals. We assume that $M \in \R^{d \times d}$ and $H \in \R^{k \times d}$ are full rank. 
    \begin{proposition}
    If $r\to 0$, then we have
        \begin{align*}
        &\rhoo^{(r)}<\infty,\\
        &\rhos^{(r)}\sim\bigO(r^{-k}).
        \end{align*}
    \end{proposition}
    \begin{proof}
    Consider the two inverse problems that correspond to $\muo^{(r)},\pio^{(r)}$ and $\mus^{(r)},\pis^{(r)}$. Note that the two problems have identical prior and design matrix. Let $\Gammao^{(r)}$ and $\Gammas^{(r)}$ denote the noise in those two inverse problems. When $r$ goes to $0$, we observe that
        \begin{align*}
            &\Gammao^{(r)}= r^2R +HQH'\to HQH',\\
            &\Gammas^{(r)}= r^2R \to 0.
        \end{align*}
    Therefore, the limit of $\rhoo^{(r)}$ converges to a finite value, but Lemma \ref{prop:noise_scaling} implies that $\rhos^{(r)}$ diverges at rate $\bigO(r^{-k})$.
    \end{proof}

    \subsection{Standard and Optimal Proposal in High Dimension}\label{ssec:dalargedimension}
    The previous subsection shows that the standard and optimal proposals can have dramatically different behavior in the small noise regime $r \to 0.$ Here we show that both proposals can also lead to dramatically different behavior in high dimensional limits. Precisely, as a consequence of Corollary \ref{cor: dimensionality 3 equivalence} we can easily identify the exact regimes where both proposals converge or diverge in limit. The notation is analogous to that in Subsection \ref{ssec:dalargedimension}, and so we omit the details. 
    
    \begin{proposition}
    Consider the sequence of particle filters defined as above. We have the following convergence criteria:
        \begin{enumerate}
            \item $\muo^{(1:\infty)}\ll\pio^{(1:\infty)}$ and $\rhoo< \infty$ if and only if $\sumi\frac{h_i^2m_i^2p_i^2}{h_i^2q_i^2+r_i^2}<\infty$,
            \item $\mus^{(1:\infty)}\ll\pis^{(1:\infty)}$ and $\rhos <\infty$ if and only if $\sumi\frac{h_i^2m_i^2p_i^2}{r_i^2}<\infty$ and 
            $\sumi\frac{h_i^2q_i^2}{r_i^2}<\infty$.
        \end{enumerate}
    \end{proposition}
    \begin{proof}
    By direct computation, we have
        \begin{align*}
        &\lambdas^{(i)} = \frac{h_i^2m_i^2p_i^2+h_i^2q_i^2}{r_i^2}=\frac{h_i^2m_i^2p_i^2}{r_i^2}+\frac{h_i^2q_i^2}{r_i^2},\\
        &\lambdao^{(i)} = \frac{h_i^2m_i^2p_i^2}{h_i^2q_i^2+r_i^2}.
        \end{align*}
   Theorem \ref{cor: dimensionality 3 equivalence} gives the desired result.
    \end{proof}
    
    \begin{example}
  As a simple example where absolute continuity holds for the optimal proposal but not for the standard one, let $h_i = m_i = p_i = r_i =1.$ Then $\rhos = \infty,$ but $\rhoo < \infty$ provided that $\sum_{i=1}^\infty \frac{1}{q_i^2 +1} < \infty. $  \qed
    \end{example}

\section*{Acknowledgement}
The work of DSA was supported by NSF and NGA through the grant DMS-2027056. DSA also acknowledges partial support from the NSF Grant DMS-1912818/1912802.
\bibliographystyle{plain} 
\bibliography{reference}

\appendix

\section{\texorpdfstring{$\chi^2$}{chi2}-divergence between Gaussians}\label{appendixA}

    We recall that the distribution $P_\theta$ parameterized by $\theta$  belongs to the exponential family $\mathcal{E}_F(\Theta)$ over a natural parameter space $\Theta$, if $\theta\in\Theta$ and $P_\theta$ has density of the form
        \[f(u;\theta)=e^{\langle t(u),\theta\rangle-F(\theta)+k(u)},\]
    where the natural parameter space is given by
        \[\Theta = \left\{\theta:\int e^{\langle t(u),\theta\rangle+k(u)}du<\infty\right\}.\]

  The following result can be found in \cite{nielsen2013chi}.    
    
	\begin{lemma}\label{lemma:chi2_expfamily}
		Suppose $\theta_{1,2}\in\Theta$ are parameters for probability densities $f(u;\theta_{1,2})=e^{\langle t(u),\theta_{1,2}\rangle-F(\theta_{1,2})+k(u)}$ with $2\theta_1-\theta_2\in\Theta.$ Then,
		\[\dchi \bigl(f(\cdotp;\theta_1)\|f(\cdotp;\theta_2) \bigr)=e^{F(2\theta_1-\theta_2)-2F(\theta_1)+F(\theta_2)}-1.\]
	\end{lemma}

	\begin{proof}
	By direct computation,
		\begin{align*}
		\dchi \bigl(f(\cdotp;\theta_1)\|f(\cdotp;\theta_2)\bigr)+1
		=&\int f(u;\theta_1)^2f(u;\theta_2)^{-1}du\\
		=&\int e^{\langle t(u),2\theta_1-\theta_2\rangle-(2F(\theta_1)-F(\theta_2))+k(u)} \, du\\
		=&e^{F(2\theta_1-\theta_2)-2F(\theta_1)+F(\theta_2)}\int f(u;2\theta_1-\theta_2) \, du\\
		=&e^{F(2\theta_1-\theta_2)-2F(\theta_1)+F(\theta_2)}.
		\end{align*}
    Note that $\int f(u;2\theta_1-\theta_2) \,du = 1$ since $2\theta_1-\theta_2 \in \Theta$ by assumption. 
	\end{proof}
    Using Lemma \ref{lemma:chi2_expfamily} we can compute the $\chi^2$-divergence between Gaussians. To do so, we note that $d-$dimensional Gaussians $\mathcal{N}(\mu,\Sigma)$ belong to the exponential family over the parameter space $\mathbb{R}^d\bigoplus\mathbb{R}^{d\times d}$ by letting $\theta = [\Sigma^{-1}\mu;-\frac{1}{2}\Sigma^{-1}]$ and  $F(\theta)=\frac{1}{2}\mu'\Sigma^{-1}\mu+\frac{1}{2}\log|\Sigma|$. In the context of Gaussians, an exponential parameter $\theta = [\Sigma^{-1}\mu;-\frac{1}{2}\Sigma^{-1}]$ belongs to the natural parameter space $\Theta$ if and only if $\Sigma$ is symmetric and positive  definite. Indeed, the integral $\int \exp(-\frac{1}{2}(u-\mu)'\Sigma^{-1}(u-\mu))du$ is finite if and only if $\Sigma \succ 0$.

    \begin{proof}[Proof of Proposition \ref{prop:gaussian_chi2}] 
    Let $\theta_\mu,\theta_\pi$ be the exponential parameters of $\mu,\pi$. Then $2\theta_\mu-\theta_\pi$ corresponds to a Gaussian with mean $(2C^{-1}-\Sigma^{-1})^{-1}(2C^{-1}m)$ and covariance $(2C^{-1}-\Sigma^{-1})^{-1}.$ We have
        \begin{align*}
        F(2\theta_\mu-\theta_\pi)-2F(\theta_\mu)+F(\theta_\pi)
        &=\frac{1}{2}\log|(2C^{-1}-\Sigma^{-1})^{-1}|-\log|C|+\frac{1}{2}\log|\Sigma| + \\
         & \quad \quad \quad \quad\frac{1}{2}(2C^{-1}m)'(2C^{-1}-\Sigma^{-1})^{-1}(2C^{-1}m)-m'C^{-1}m\\
        &=\log\sqrt{\frac{|\Sigma|}{|2C^{-1}-\Sigma^{-1}||C|^2}}
        + m'(C^{-1}(2C^{-1}-\Sigma^{-1})^{-1}2C^{-1})m\\
            &\quad \quad \quad \quad-m'(C^{-1}(2C^{-1}-\Sigma^{-1})^{-1}(2C^{-1}-\Sigma^{-1})) m\\
        &=\log\frac{|\Sigma|}{\sqrt{|2\Sigma-C||C|}}
        +m'(C^{-1}(2C^{-1}-\Sigma^{-1})^{-1}\Sigma^{-1})m\\
        &=\log\frac{|\Sigma|}{\sqrt{|2\Sigma-C||C|}}
        +m'(2\Sigma-C)^{-1}m.
        \end{align*}
    Applying Lemma \ref{lemma:chi2_expfamily} gives 
        \begin{align*}
        \dchi(\mu\|\pi)
        &=\exp\Bigl( F(2\theta_\mu-\theta_\pi)-2F(\theta_\mu)+F(\theta_\pi) \Bigr)-1\\
        &=\frac{|\Sigma|}{\sqrt{|2\Sigma-C||C|}} \exp \Bigl(m'(2\Sigma-C)^{-1}m \Bigr) -1,
        \end{align*}
    if $2\theta_\mu-\theta_\pi\in\Theta$. In other words, the corresponding covariance matrix $(2C^{-1}-\Sigma^{-1})^{-1}$ is positive  definite. 
    \end{proof}
    
    \begin{remark}
    By translation invariance of Lebesgue measure, we can obtain the more general formula for $\chi^2$-divergence between two Gaussians with non-zero mean by replacing $m$ with the difference between the two mean vectors:
        \[\dchi \Bigl(\Nc(m_1,C)\| \, \Nc(m_2,\Sigma)\Bigr)
        = \frac{|\Sigma|}{\sqrt{|2\Sigma-C||C|}}e^{(m_1-m_2)' (2 \Sigma - C)^{-1}(m_1-m_2)}  -1 .\]
    \end{remark}

\section{Proof of Lemma \ref{lem: hellinger integral abs cts}}\label{aappendixB}
\begin{proof}
    Dividing $g$ by its normalizing constant, we may assume without loss of generality that $g$ is exactly the Radon-Nikodym derivative $\frac{d\mu}{d\pi}$ and $\mathcal{H}(\mu,\pi)=\pi_i(\sqrt{g})$.\\
	If $\mu_{1:\infty}\ll\pi_{1:\infty}$, then the Radon-Nikodym derivative $g_{1:\infty}$ cannot be $\pi_{1:\infty}$ a.e. zero since $\pi_{1:\infty}$ and $\mu_{1:\infty}$ are probability measures. As a consequence, $\prodi\pi_i\left(\sqrt{g_i}\right)=\pi_{1:\infty}\left(\sqrt{g_{1:\infty}}\right)>0$ by the product structure of $\mu_{1:\infty}$ and $\pi_{1:\infty}$.\\
	Now we assume $\prodi\pi_i\left(\sqrt{g_i}\right)>0$. It suffices to show that $g_{1:\infty}$ is well-defined, i.e. convergence of $\prod_{i=1}^Lg_i$ in $L_{\pi}^1$ as $L\to\infty$. It suffices to prove that the sequence is Cauchy, in other words
		\begin{align*}
			\lim_{L,\ell\to\infty}\pi_{1:\infty}\left(|g_{1:L+\ell}-g_{1:L}|\right)=0.
		\end{align*}
	We observe that
		\begin{align*}
			\|g_{1:L+\ell}-g_{1:L}\|_1
			&\leq \|\sqrt{g_{1:L+\ell}}-\sqrt{g_{1:L}}\|_2\|\sqrt{g_{1:L+\ell}}+\sqrt{g_{1:L}}\|_2\\
			&\leq \|\sqrt{g_{1:L+\ell}}-\sqrt{g_{1:L}}\|_2(\|\sqrt{g_{1:L+\ell}}\|_2+\|\sqrt{g_{1:L}}\|_2)\\
			&= 2\|\sqrt{g_{1:L+\ell}}-\sqrt{g_{1:L}}\|_2.
		\end{align*}
    Expanding the square of the right-hand side gives 
		\begin{align*}
			\pi_{1:\infty}\left(\left|\sqrt{g_{1:L+\ell}}-\sqrt{g_{1:L}}\right|^2\right)
			&=\pi_{1:\infty}\left(g_{1:L+\ell}+g_{1:L}-2\sqrt{g_{1:L+\ell}g_{1:L}}\right)\\
			&=2-2\pi_{1:L}\left(g_{1:L}\right)\pi_{L+1:\infty}\left(\sqrt{\frac{g_{1:L+\ell}}{g_{1:L}}}\right)\\
			&=2\left(1-\frac{\pi_{1:L+\ell}\left(\sqrt{g_{1:L+\ell}}\right)}{\pi_{1:L}\left(\sqrt{g_{1:L}}\right)}\right).
		\end{align*}
	Therefore, it is enough to show
		\begin{align*}
			\lim_{L,\ell\to\infty}\frac{\pi_{1:L+\ell}\left(\sqrt{g_{1:L+\ell}}\right)}{\pi_{1:L}\left(\sqrt{g_{1:L}}\right)}=1.
		\end{align*}
	By Jensen's inequality, for any two probability measures $\mu\ll\pi$ with density $g$, we have
		\begin{align}\label{ineq: hellinger bounded}
		\pi\left(\sqrt{g}\right)\leq \sqrt{\pi\left({g}\right)} = 1.
		\end{align}
	Combining with our assumption, we deduce that
	\[0<\prodi\pi_i\left(\sqrt{g_i}\right)=\pi_{1:\infty}\left(\sqrt{g_{1:\infty}}\right)\leq 1,\]
	which is equivalent to
	\[-\infty< \sumi\log(\pi_i\left(\sqrt{g_i}\right))\leq 0.\]
	This series is monotonely decreasing by \eqref{ineq: hellinger bounded} and bounded below, so it converges and  satisfies that
		\begin{align*}
			\lim_{L,\ell\to\infty}\frac{\pi_{1:L+\ell}\left(\sqrt{g_{1:L+\ell}}\right)}{\pi_{1:L}\left(\sqrt{g_{1:L}}\right)}
			=\lim_{L,\ell\to\infty}e^{\sum_{i=L}^{L+\ell}{\log(\pi_i\left(\sqrt{g_i}\right))}}=1.
		\end{align*}
	\end{proof}

\newpage

\section{Additional Figures}\label{AppendixC}

\begin{figure}[!htb] 
    \begin{minipage}{.5\textwidth}
       
        \includegraphics[width=1\linewidth, height=0.21\textheight]{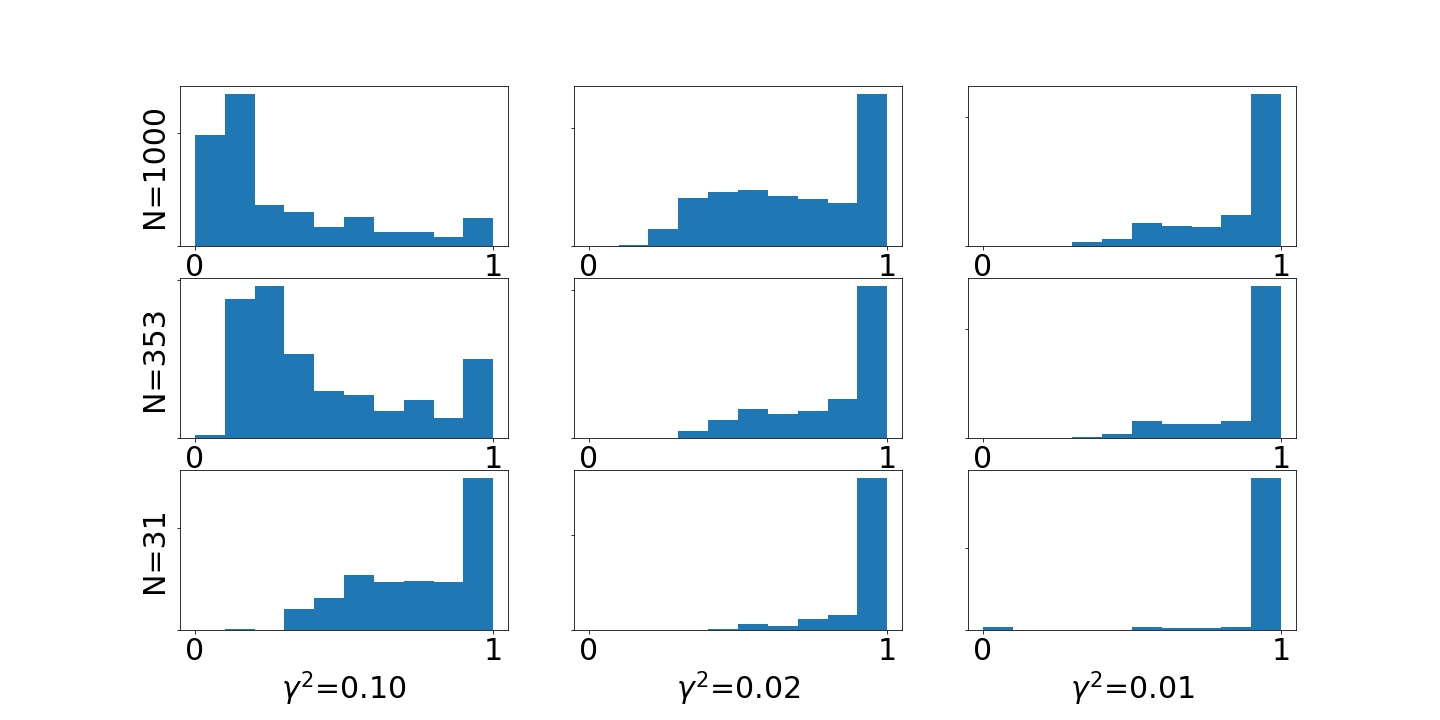}
        \subcaption{$N=\gamma^{-3}$.}
        \label{fig:prob1_6_2}
    \end{minipage}%
    \begin{minipage}{0.5\textwidth}
       
        \includegraphics[width=1\linewidth, height=0.21\textheight]{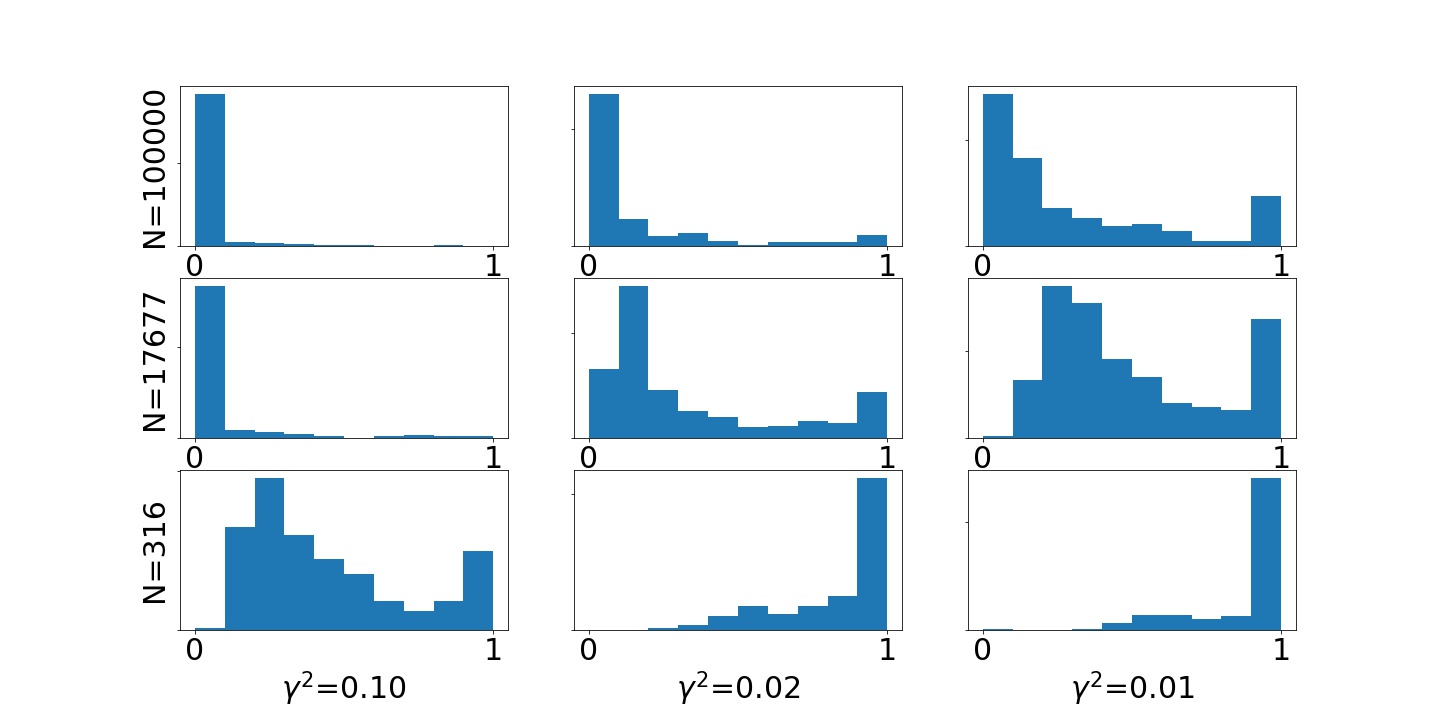}
        \subcaption{$N = \gamma^{-5}$.}
        \label{fig:prob1_6_1}
    \end{minipage}
    \caption{Noise scaling with $d = k =  4.$ \label{fig:Noise Scaling d=4}}
\end{figure}

\begin{figure}[!htb] 
    \begin{minipage}{.5\textwidth}
       
        \includegraphics[width=1\linewidth, height=0.21\textheight]{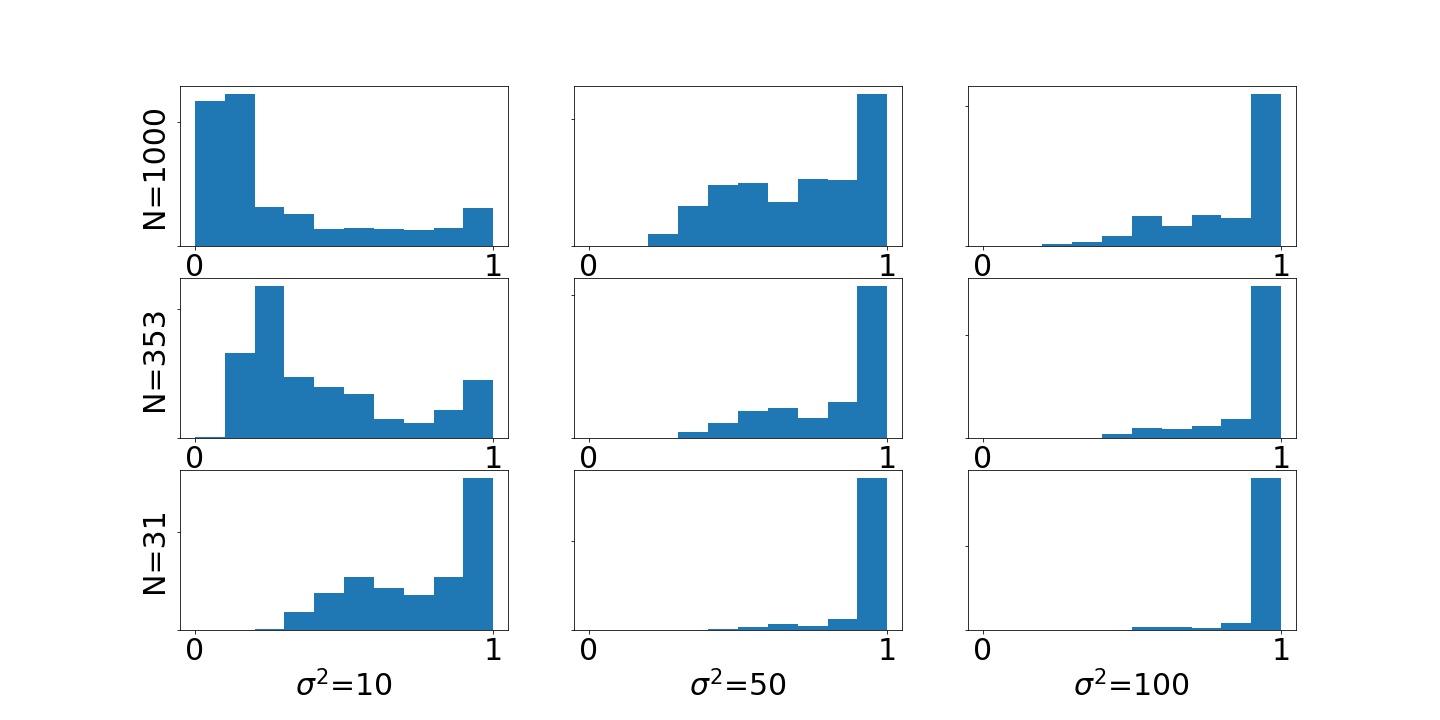}
        \subcaption{$N=\sigma^{-3}$.}
        \label{fig:prob1_6_2}
    \end{minipage}%
    \begin{minipage}{0.5\textwidth}
       
        \includegraphics[width=1\linewidth, height=0.21\textheight]{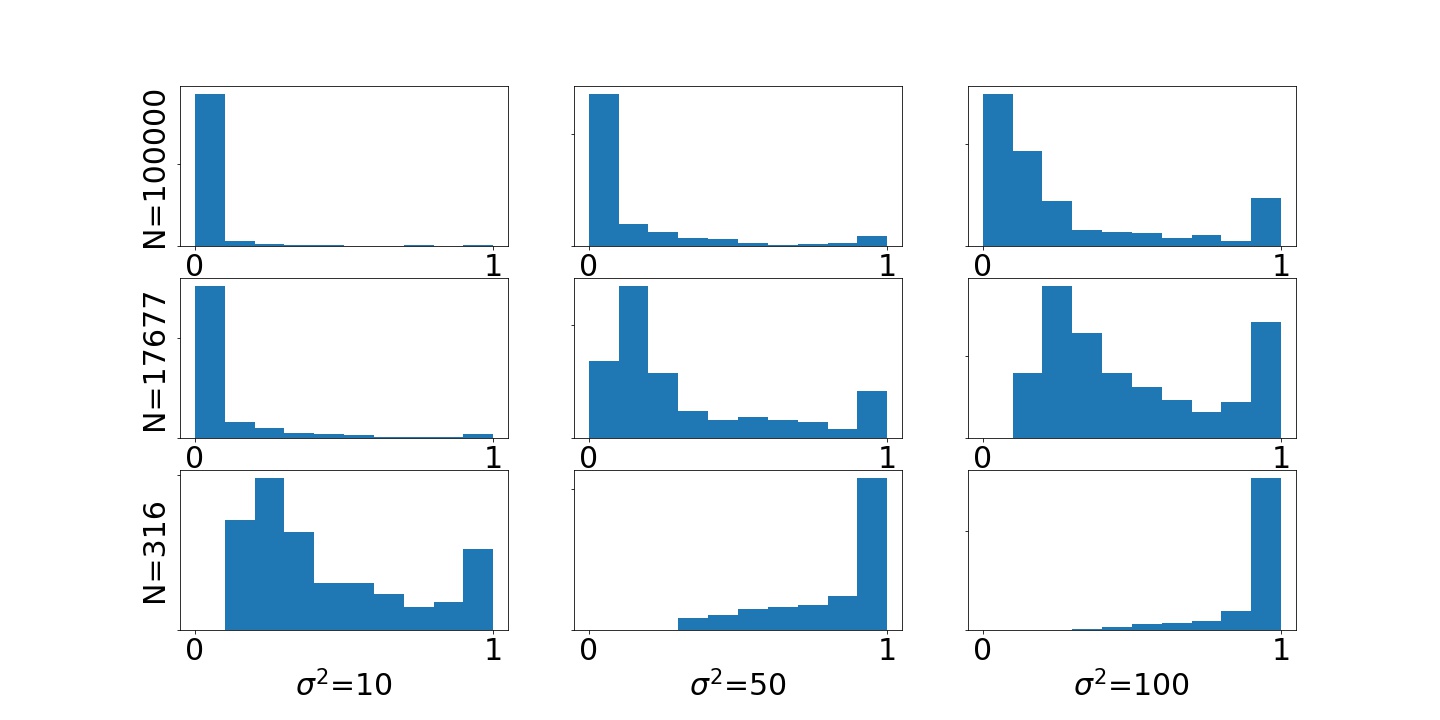}
        \subcaption{$N = \sigma^{-5}$.}
        \label{fig:prob1_6_1}
    \end{minipage}
    \caption{Prior scaling with $d = k =  4.$ \label{fig:Prior Scaling d=4}}
\end{figure}

\begin{figure}[!htb] 
    \begin{minipage}{.5\textwidth}
       
        \includegraphics[width=1\linewidth, height=0.21\textheight]{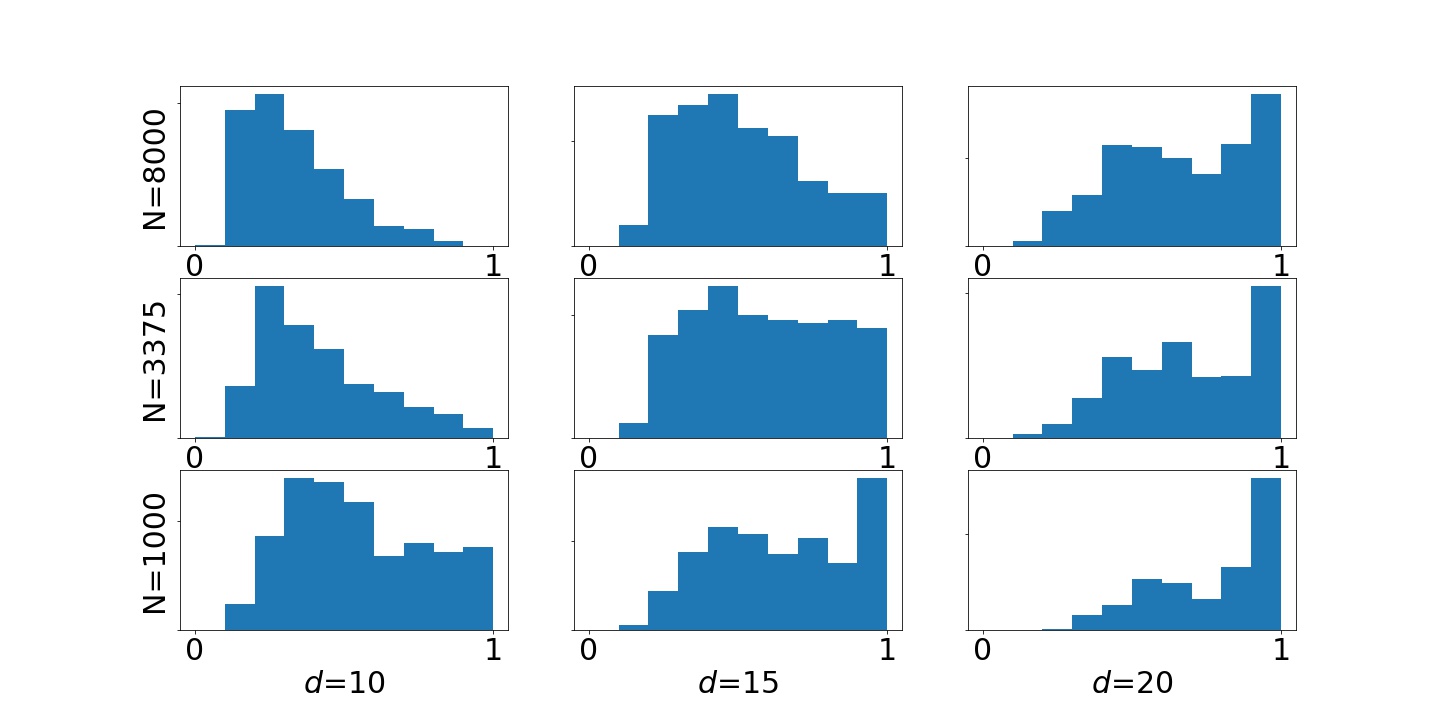}
        \subcaption{$N=d^3$.}
        \label{fig:dim no collapse}
    \end{minipage}%
    \begin{minipage}{0.5\textwidth}
       
        \includegraphics[width=1\linewidth, height=0.21\textheight]{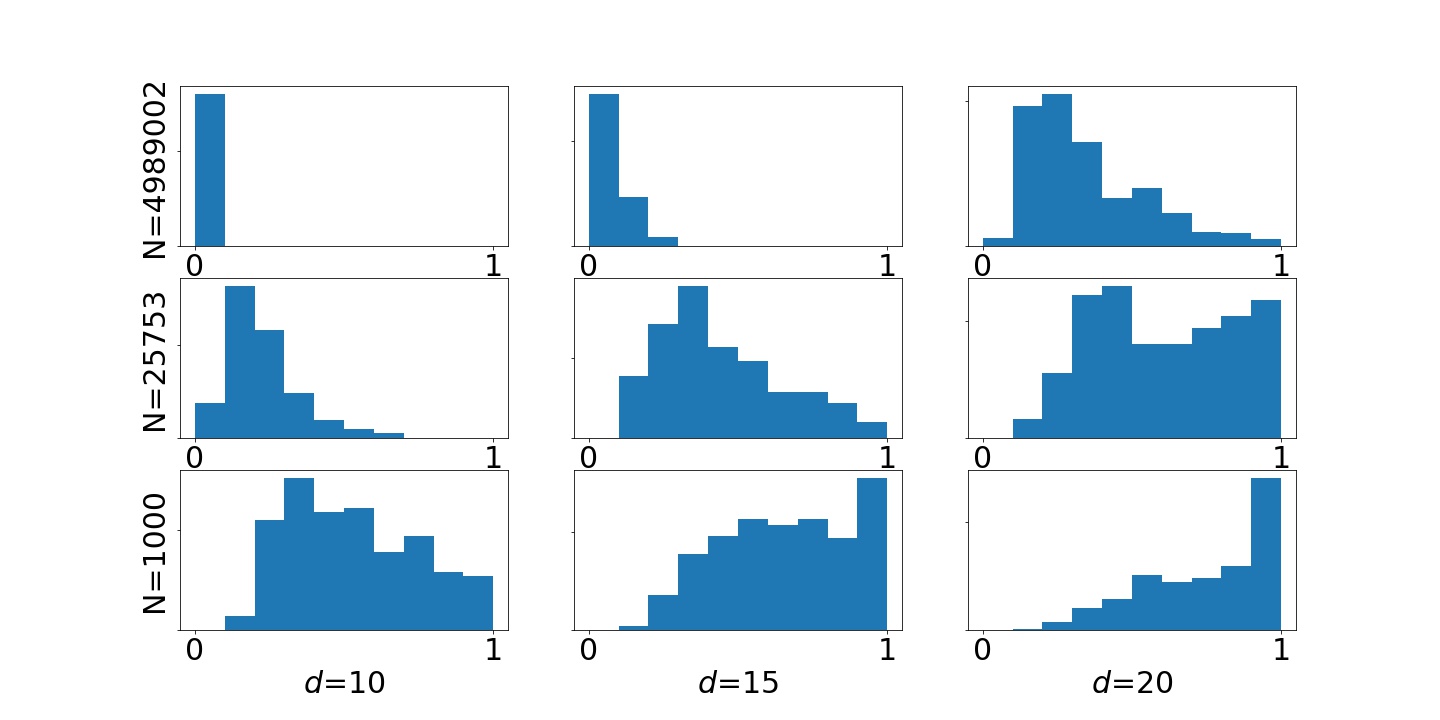}
        \subcaption{$N=\bigO \bigl(\dchi(\mu_{1:d}\|\pi_{1:d}) \bigr)$.}
        \label{fig:dim collapse}
    \end{minipage}
    \caption{Dimensional scaling $\lambda = 2.4.$ \label{fig:Dimensional Scaling lambda=2.4}}
\end{figure}

\end{document}